\documentclass[a4paper]{article}

\usepackage[shortlabels]{enumitem}
\usepackage[]{geometry}
\usepackage[utf8]{inputenc}
\usepackage{fancybox, calc}
\usepackage{amssymb}
\usepackage{amsmath}
\usepackage{hyperref}
 \usepackage{systeme}
\usepackage{graphicx}
\usepackage{enumitem}
\usepackage{ucs}
\usepackage{amsfonts}
\usepackage{amssymb}
\usepackage{makeidx}
\usepackage{diagbox}
\usepackage{tabularx}
\usepackage{float}
\usepackage[usenames, dvipsnames]{color}
\usepackage{url}
\usepackage{marvosym}
\usepackage{wasysym}
\usepackage{mathrsfs}
\usepackage[utf8]{inputenc}
\usepackage[english]{babel}
\usepackage[normalem]{ulem}
\usepackage{ucs}
\usepackage{amsfonts}
\usepackage{dcolumn}
\usepackage{makeidx}
\usepackage{bm}
\usepackage[usenames]{color}
\usepackage{multirow}
\usepackage{graphicx}
\usepackage{hyperref}
\usepackage{tikz}
\usepackage{subfig}
\usepackage{booktabs}
\usepackage{xcolor}
\usepackage{fancyhdr}
\usepackage{amsmath, amsthm, amssymb}
\usepackage{multicol}
\usepackage{yfonts}
\usepackage{multicol}
\usepackage{caption}
\usepackage{cancel}
\usepackage{amssymb,amsfonts,amssymb,bbm}
\usepackage{amsmath}
\usepackage{phaistos,hieroglf,staves,ifsym,marvosym,skull}
\usepackage{tikz}
\usepackage{transparent}
\usepackage{eso-pic}
\usepackage{soul}
\usepackage{lettrine}
\usepackage{mathrsfs}
\date{\today}

\def\Rset{\mathbb{R}}
\def\pdf{f}
\def\supp{\Omega}
\newcommand{\down}[1]{\mathfrak {D}_{#1}}
\newcommand{\adown}{{\pdf}^\downarrow_{\alpha}}
\newcommand{\xidown}{{\pdf}^\downarrow_{\xi}}
\newcommand{\up}[1]{\mathfrak {U}_{#1}}
\newcommand{\aup}{{\pdf}^\uparrow_{\alpha}}

\def\Rset{\mathbb{R}}

\def\esssup{\operatorname{esssup}}
\def\argmax{\operatorname*{argmax}}
\def\sign{\operatorname{sign}}
\def\D{\mathcal{D}}

\def\gauss{g}

\def\pdf{f}

\newcommand{\descort}[2][]{\mathfrak{E}_{#1}\ifthenelse{\isempty{#2}}{}{ {\left[ #2 \right]}}}
\newcommand{\escort}[2][]{\varepsilon_{#1}\ifthenelse{\isempty{#2}}{}{ {\left[ #2 \right]}}}

\newtheorem{theorem}{Theorem}[section]
\newtheorem{corollary}{Corollary}[section]

\newtheorem{lemma}{Lemma}[section]
\newtheorem{definition}{Definition}[section]
\newtheorem{proposition}{Proposition}[section]

\newtheorem{remark}{Remark}[section]
\usepackage{cite}

\numberwithin{equation}{section}

\definecolor{rougeG}{rgb}{.76,0,.12}
\definecolor{vertG}{rgb}{.07,.56,.25}

\title{Through and beyond moments, entropies and Fisher information measures: new informational functionals and inequalities}

\author{ {\bf Razvan Gabriel Iagar\footnotemark[1]\footnote{\textit{e-mail}: razvan.iagar@urjc.es} and David Puertas-Centeno\footnotemark[2]\footnote{\textit{e-mail}: david.puertas@urjc.es}} \\ \\
	Departamento de Matem\'atica Aplicada, \\ Ciencia e Ingeniería de los Materiales y Tecnología Electrónica\\
	Universidad Rey Juan Carlos \\
	C/~Tulip\'an s/n, 28933 M\'ostoles \\
	Madrid, Spain}

\usepackage{yfonts}

\usepackage{blindtext}

\usepackage[T1]{fontenc}
\usepackage{palatino}
\begin{document}
	\maketitle

\begin{abstract}

We introduce new classes of informational functionals, called \emph{upper moments}, respectively \emph{down-Fisher measures}, obtained by applying classical functionals such as $p$-moments and the Fisher information to the recently introduced up or down transformed probability density functions. We extend some of the the most important informational inequalities to our new functionals and establish optimal constants and minimizers for them. In particular, we highlight that, under certain constraints, the generalized Beta probability density maximizes (or minimizes) the upper-moments when the moment is fixed. Moreover, we apply these structured inequalities to systematically establish new and sharp upper bounds for the main classical informational products such as moment-entropy, Stam, or Cramér-Rao like products under certain regularity conditions. Other relevant properties, such as regularity under scaling changes or monotonicity with respect to the parameter, are studied. Applications to related problems to the Hausdorff moment problem are also given.

\end{abstract}

\section{Introduction}

Over the last decades, information theory had regain great interest among mathematicians and physicists, related to the study of generalized informational functionals and measures and their applications to a wide amount of fields. In particular, extensions of the Shannon entropy such as the R\'enyi and the Tsallis entropies (which are one-to-one mapped) have been found useful in a number of different areas of knowledge.  Just to give a few examples of such applications, we mention that the R\'enyi entropy is related to fractal dimensions and multifractal analysis~\cite{Badii87,Chhabra89, Beck93, Harte01, Chen20,Jizba04,Jizba04b}, signal analysis~\cite{Vignat03} and it has been recently applied to the analysis of complex time series~\cite{Contreras-Reyes23}, while the Tsallis entropy~\cite{Tsallis22} has been shown to be the only composable (that is, the entropy of a compound system can be calculated in terms of the entropy of its independent components) entropy function with trace form, in the sense described by ~\cite{Enciso17}. In addition, the Tsallis and Rényi entropies have applications in a wide range of areas of science, from non-extensive statistical mechanics to high energy physics, among others~\cite{Tsallis22}. On the one hand, focusing on the continuous case, very interesting applications of the R\'enyi entropy to the uncertainty inequalities for position and momentum in quantum mechanics have been established in \cite{Bialynicki84, Bialynicki06, Zozor08}. On the other hand, applications of the Fisher information as well as the so-called statistical complexity products have been established in the study of electronic and nuclear structure~\cite{Yanez94,Romera04,Sen11,Dehesa23,Angulo23, Dehesa24}, in the classification of more complex molecular structures~\cite{Esquivel10,LopezRosa16}, in the analysis of ionization processes and isoelectronic series~\cite{SenAngulo07}, in the study of confined quantum systems\cite{Estanon24}, and for quantifying chemical reactivity properties~\cite{Rong20}. 

A number of informational inequalities between entropies, moments and other functionals or measures such as the Fisher information have been obtained and their optimal constants and minimizers have been also established. Classical inequalities, as described in \cite{Dembo91}, have the Gaussian function as minimizer and follow from well known functional inequalities such as Young, Brunn-Minkowski or Jensen's inequality. More recently, three of the most important informational inequalities, namely the moment-entropy inequality, the Stam inequality and the Cram\'er-Rao inequality, have been extended to involve the R\'enyi entropy, the $p$-th moment of a probability density function and the generalized Fisher information in \cite{Lutwak05, Bercher12}, and power-like functions known as the $q$-Gaussians or stretched deformed Gaussians replaced the classical Gaussian in the role of minimizers of these extended inequalities.

A different range of applications of the Shannon entropy and its above mentioned generalizations has been found in connection to the reconstruction problem, that is, the quest for mathematically characterizing and then numerically computing a probability density function starting from the knowledge of a (finite or infinite) part of its moments. The Hamburger, Stieltjes and Hausdorff moment problems consist in determining the conditions under which a sequence of numbers uniquely characterizes a probability density, when the respective supports of the density function are $\Rset,\Rset^+$ and the interval $[0,1]$ respectively, and, in such case, find the corresponding probability density. Such reconstruction problems are known to be numerically unstable \cite{Talenti1987} and mathematically ill-posed, see for example \cite{Shohat1970, Lasserre2010}. In order to approach these difficulties, researchers developed the MaxEnt principle \cite{Jaynes57, Mnatsakanov2008, Biswas2010}, which amounts to give approximate solutions to the reconstruction problems by maximizing an entropy functional, with the constraint of a finite number of fixed moments. However, a drawback of this approach appears when dealing with heavy-tailed density functions (if their support is either $\Rset$ or $\Rset^+$) or, even more strikingly, with compactly supported densities that might be divergent at the boundary of their compact support (in the case of the Hausdorff moment problem). In all these cases, only a few moments (or none of them) are finite and thus the MaxEnt approach is rather limited.

One way to remove the previous restrictions and to widen the class of applications of the entropy functionals is to consider transformations acting on the set of probability density functions and mapping density functions with bad integrability properties (such as heavy tails or even divergence at the boundary of the support) into new density functions with improved features in the sense of well defined moments, entropies or other informational functionals. Just as an example, one of the authors introduced in \cite{Puertas19, Zozor17, Puertas25} the differential-escort transformation and the cumulative moments and deduced new informational inequalities saturated by generalized trigonometric functions.

A new and different step in this direction has been proposed by the authors in their previous work \cite{IP25} by considering two mutually-inverse transformations, called \emph{up} and \emph{down}, and exploring their properties. It has been noticed that these transformations uncover an interesting structure of relations between the informational functionals (such as the $p$-th moments, the R\'enyi entropy power and the generalized Fisher information) of a probability density function with respect to similar functionals applied to its up transformed or down transformed density, allowing to extend such functionals to densities with heavy tails as $|x|\to\infty$ or compactly supported densities that are divergent (but still with good integrability properties) at the edges of their support. Applications to the Hausdorff moment problem and to the MaxEnt principle are described and further generalizations of the Stam and moment-entropy inequalities to mirrored ranges of exponents and parameters are given therein. A remarkable aspect of these generalized informational inequalities obtained in \cite{IP25} is that they are saturated by densities that can be divergent at the edge of their support, strongly extending in this way the class of densities that can be considered for this approach.

In the present paper, we continue the analysis of the up and down transformations and their applications started in the above mentioned previous paper \cite{IP25} by introducing informational functionals whose "natural" deduction comes from an application of the classical functionals to densities obtained by a successive application of the up and down transformations a finite number of times to other density functions. We thus introduce the \emph{upper-moments} of order $n$, for any natural number $n\geq2$, obtained as the $p$-th moment of an up transformed density, and the \emph{down-Fisher measure}, obtained as the generalized Fisher information applied to a down transformed density. We also deduce sharp informational inequalities involving these newly introduced functionals, thus relating the amount of information contained in these functionals to the classical ones. Finally, some applications to the Hausdorff moment problem and the MaxEnt approach are given, in the form of establishing an equivalence between the moment problem and an upper-moment problem and extending the MaxEnt principle to maximizing upper-moments of higher order under the constraint of a finite number of fixed upper-moments of lower order.

Over the last decades, further extensions of the notion of entropy have been developed, such as the Kaniadakis entropy~\cite{Kaniadakis09}, group entropies~\cite{Tempesta19}, or $\phi$-entropies~\cite{Toranzo17}, to name but a few of them. It is worth mentioning that all the quantities involved throughout the paper can be understood as entropic functionals, since they correspond to the entropy of some transformed probability denstity. Throughout the paper, the new functionals are applied in order to establish upper bounds for the classical informational products (also called complexity products~\cite{GuerreroDehesa11}) such as moment-entropy, Stam, or Cram\'er-Rao, under certain regularity conditions of the probability density. This fact gains great potential interest due to the large number of applications of these classical quantities.

We highlight that the upper-moments and the moments are related in the same way as the moments and Rényi entropy power are, but also in the same way as the Rényi entropy power and the generalized Fisher measures, or the Fisher and the down-Fisher measures respectively.
\medskip

\textbf{Structure of the paper.} We start with a section of preliminary facts, recalling, on the one hand, theoretical aspects of classical moments and entropies and, on the other hand and for the sake of completeness, the definition and most important properties of the up and down transformations introduced in \cite{IP25}. All the results that are used in the forthcoming proofs are thus recalled in Section \ref{sec:prelim}. The new informational functionals mentioned in the previous paragraph are then defined and briefly described in Section \ref{sec:newinf}, where their connection to the up and down transformations is explained and the applications to reconstruction problems and to the MaxEnt approach explained above are rigorously stated. Section \ref{sec:ineq} is what we consider to be the most interesting contribution of the paper. In this section, we prove a number of new informational inequalities involving the informational functionals introduced in the previous section, establishing optimal ranges of parameters and, when possible, optimal constants and minimizers for them. More precisely, we state and prove three different types of inequalities, as described below:

$\bullet$ Inequalities relating the upper-moments to the classical moments, or, in a general form, upper-moments of higher order to upper-moments of lower order, are proved in Section \ref{sec:upmineq}.

$\bullet$ Inequalities relating the upper-moments to the entropy power, which are generalizations of the moment-entropy inequality, are stated and proved in Section \ref{sec:upment}.

$\bullet$ Inequalities relating the down-Fisher measure to the generalized Fisher information, on the one hand, and to the Shannon entropy, on the other hand, are stated and proved in Sections \ref{sec:dffineq} and \ref{sec:dfshineq}.

We are now in a position to formulate the mathematical statements of the above mentioned functionals and inequalities, after recalling the main tools and classical results which we shall employ in the statements and proofs.

\section{Preliminary notions and results}\label{sec:prelim}

We collect in this section, for the sake of completeness, a number of notions, definitions and results that will be employed throughout the paper.

\subsection{A brief review of moments and entropies}

We recall below a number of definitions and well established properties of the most standard informational functionals used in the present work. We follow closely the structure of the similar section in our previous paper \cite{IP25}.

\medskip

\noindent \textbf{The $p$-th absolute moment} of a probability density function $\pdf$, with $p \geqslant 0$, is given by
\begin{equation}\label{eq:pabs}
\mu_p[\pdf] = \int_\Rset |x|^p \, \pdf(x) \, dx  = \left\langle  \, |x|^p \, \right\rangle_\pdf,
\end{equation}
when the integral converges. Associated to the $p$-the absolute moment is the $p$-typical deviation defined as described below:
\begin{equation*}
\begin{split}
&\sigma_p[\pdf] =\left(\int_\Rset |x|^p \, \pdf(x) \, dx \right)^{\frac{1}{p}}, \quad {\rm for} \ p > 0, \\
&\sigma_0[\pdf] = \lim\limits_{p \to 0} \sigma_p[\pdf]  = \exp\left(\int_\Rset \pdf(x) \, \log|x| \, dx \right), \\
&\sigma_{\infty}[\pdf]=\lim\limits_{p\to\infty}\sigma_p[\pdf]=\esssup\big\{|x| : x\in\Rset, \pdf(x) > 0 \big\}.
\end{split}
\end{equation*}
In other words, for $p>0$, the $p$-typical deviation is the $p$-th absolute moment taken at the power $1/p$ (with the corresponding limits as $p\to0$ and as $p\to\infty$). Although in the classical theory $\mu_p$ and $\sigma_p$ are not defined for exponents $p<0$, we may consider in this work negative moments as well, for density functions for which they are well defined. We also recall the definition of the logarithmic moments
\begin{equation}\label{eq:sigmaL}
\sigma_p^{(L)}[f]=\int_{\Rset} f(x)|\log |x||^pdx,
\end{equation}
and of the exponential moments following~\cite{Mnatsakanov2013}
\begin{equation}\label{eq:sigmaE}
\sigma_{p}^{(E)}[f]=\left(\int_{\Rset} e^{-px} f(x)dx\right)^{\frac{1}{p}}=\left\langle e^{-p x}\right\rangle^{\frac 1{p}}_f.
\end{equation}

\medskip

\noindent \textbf{R\'enyi and Tsallis entropies.} The R\'enyi and Tsallis entropies of $\lambda$-order, $\lambda\neq1$, of a probability density function $\pdf$ defined on $\Rset$ (or on a Lebesgue measurable subset of $\Rset$) are defined as
\begin{equation*}
R_\lambda[\pdf] = \frac1{1-\lambda} \log\left( \int_\Rset [\pdf(x)]^\lambda \, dx \right), \qquad T_\lambda[\pdf] = \frac1{\lambda-1} \left( 1 - \int_\Rset [\pdf(x)]^\lambda \, dx \right),
\end{equation*}
while for $\lambda=1$ we recover the Shannon entropy
\begin{equation*}
\lim\limits_{\lambda \to 1}R_{\lambda}[\pdf]=\lim\limits_{\lambda \to 1} T_{\lambda}[\pdf]=S[\pdf]=-\int_\Rset \pdf(x) \, \log\pdf(x) \, dx.
\end{equation*}
In the proofs of the forthcoming informational inequalities, the R\'enyi entropy power, which is the exponential of the R\'enyi entropy,
$$
N_\lambda[\pdf] = e^{R_{\lambda}[\pdf]} = \left\langle\pdf^{\lambda-1}(x)\right\rangle^\frac1{1-\lambda}_\pdf,
$$
is usually employed. We will also denote by $N[\pdf]=e^{S[\pdf]}$ the Shannon entropy power. Observe that the R\'enyi and Tsallis entropies are one-to-one mapped by
$$
T_{\lambda}[\pdf]=\frac{e^{(1-\lambda) R_{\lambda}[\pdf]} - 1}{1 - \lambda},
$$
and thus any result involving the R\'enyi entropy can be readily expressed as a result for the Tsallis entropy and viceversa. This is why, we will only work with the R\'enyi entropy (and its power version) throughout the paper.

\medskip

\noindent \textbf{$(p,\lambda)$-Fisher information.} An extension of the Fisher information applicable to derivable probability density functions was introduced by Lutwak and Bercher~\cite{Lutwak05, Bercher12, Bercher12a}. Given $p>1$ and $\lambda \in \Rset^*$, the $(p,\lambda)$-Fisher information of a probability density function $f$ is defined as
\begin{equation}\label{eq:def_FI}
F_{p,\lambda}[\pdf]=\int_\Rset \left|\pdf(x)^{\lambda-2} \, \frac{d\pdf}{dx}(x)\right|^{p} \pdf(x) \, dx
\end{equation}
when $\pdf$ is differentiable on the closure of its support. Observe that the $(2,1)$-Fisher information reduces to the standard Fisher information of a probability density. In some cases, we shall employ for simplicity the related functional
\begin{equation}\label{eq:def_FI_bis}
\phi_{p,\lambda}[\pdf]=\big(F_{p,\lambda}[\pdf]\big)^{\frac{1}{p\lambda}}.
\end{equation}
Despite the fact that the Fisher information is defined, in the above mentioned references, only for exponents $p>1$, we extend in this paper the definitions \eqref{eq:def_FI} and \eqref{eq:def_FI_bis} to exponents $p<1$ and even $p<0$, when the corresponding functionals are well defined. Let us stress here that, in the mirrored domain of some informational inequalities, we arrive at probability density functions defined on a bounded interval and divergent at the edges of it. For such functions, exactly negative values of $p$ are more likely to make the expression of $F_{p,\lambda}[f]$ in Eq. \eqref{eq:def_FI} finite.

\medskip

\noindent \textbf{The bi-parametric Stam inequality.} This is an inequality establishing that the product of the Rényi entropy power and the $(p,\lambda)$-Fisher information is bounded from below. More precisely, given $p\in[1,+\infty)$, $p^*=\frac{p}{p-1}$ (with the convention $p^*=\infty$ for $p=1$) and $\lambda>\frac1{1+p^*}$, the following inequality
\begin{equation}\label{ineq:bip_Stam}
\phi_{p,\lambda}[\pdf] \, N_{\lambda}[\pdf] \: \geqslant \: \phi_{p,\lambda}[\gauss_{p,\lambda}] \, N_{\lambda}[\gauss_{p,\lambda}]\equiv K^{(1)}_{p,\lambda},
\end{equation}
holds true for any absolutely continuous probability density $\pdf$, according to~\cite{Lutwak05, Bercher12a, Zozor17}. Observe that the inequality Eq. \eqref{ineq:bip_Stam} is saturated by the minimizers $\gauss_{p,\lambda}$ appearing in the right hand side, known as generalized Gaussians \cite{Lutwak05}, $q-$Gaussians \cite{Bercher12a} or stretched deformed Gaussians \cite{Zozor17}. For $p>1$, these minimizers are given by
\begin{equation}\label{def:g_plambda}
\gauss_{p,\lambda}(x) \, = \, \frac{a_{p,\lambda}}{\exp_\lambda\left( |x|^{p^*} \right)}
\, = \, a_{p,\lambda}\, \exp_{2-\lambda}\left( - |x|^{p^*} \right),
\end{equation}
where $\exp_\lambda$ is the generalized Tsallis exponential
\begin{equation}\label{def:q-exp}
\exp_\lambda(x) = \left( 1 + (1 - \lambda) \, x \right)_+^\frac1{1-\lambda}, \ \ \lambda \ne 1, \qquad \exp_1(x) \: \equiv \: \lim_{\lambda \to 1} \, \exp_\lambda(x) \: = \: \exp(x),
\end{equation}
and $a_{p,\lambda}$ has an explicit value which we omit here (see \cite{IP25}). Observe that the definition \eqref{def:g_plambda} and the inequality \eqref{ineq:bip_Stam} can be extended to exponents $p^* \in (0,1)$ (that is, $p \in (-\infty,0)$). In the special case $p^*=p=0$ and $\lambda>1$, the Stam inequality also applies, but its minimizer is given by
$$
\gauss_{0,\lambda} = a_{0,\lambda}(-\log|x|)_+^{\frac1{\lambda-1}}, \quad a_{0,\lambda} = \frac1{2\Gamma\left(\frac{\lambda}{\lambda-1}\right)}.
$$
The limit $p\to 1$ entails $p^*\to\infty$ and then $\gauss_{1,\lambda}$ becomes a constant density over a unit length support, while in the limit $p\to\infty$ the inequality also holds true by taking instead of $\phi_{p,\lambda}^{\lambda}$ the essential supremum (that is, the $L^{\infty}$ norm) as observed in~\cite{Lutwak04}.

\medskip

\noindent \textbf{The moment-entropy inequality} is an informational inequality relating the R\'enyi power entropy and the moments $\sigma_p$. More precisely, when
\begin{equation}\label{eq:param_clas}
p^*\in[0,\infty ), \quad {\rm and} \quad \lambda>\frac1{1+p^*},
\end{equation}
it was proved in~\cite{Lutwak04, Lutwak05, Bercher12} that
\begin{equation}\label{ineq:bip_E-M}
\frac{\sigma_{p^*}[\pdf]}{N_{\lambda}[\pdf]} \, \geqslant \, \frac{\sigma_ {p^*}[\gauss_{p,\lambda}]}{N_{\lambda}[\gauss_{p,\lambda}]}\equiv K^{(0)}_{p,\lambda},
\end{equation}
and the minimizers of Eq. \eqref{ineq:bip_E-M} are the same deformed Gaussians $g_{p,\lambda}$ as for the generalized bi-parametric Stam inequality Eq.~\eqref{ineq:bip_Stam}. We extended the inequality \eqref{ineq:bip_E-M} in our previous work \cite[Theorem 5.2]{IP25} to a mirrored range of parameters. More precisely, it was established therein that, if
\begin{equation}\label{eq:param_mir}
\lambda<0,\quad\sign\left(\frac{\lambda-1}{ \lambda}+ p^*\right)=\sign\left(1-p^*\right),
\end{equation}
then, for any continuously differentiable density function $f$, the following mirrored moment-entropy inequality holds true:
\begin{equation}\label{ineq:bip_E-M_mirrored}
\left(\frac{\sigma_{p^*}[\pdf]}{N_{\lambda}[\pdf]}\right)^{p^*-1}\geqslant\left(\frac{\sigma_{p^*}[\overline g_{p,\lambda}]}{N_{\lambda}[\overline g_{p,\lambda}]}\right)^{p^*-1}\equiv \kappa^{(0)}_{p,\lambda},
\end{equation}
where $\overline g_{p,\lambda}=g_{1-\lambda,1-p}$ and $\kappa^{(0)}_{p,\lambda}=|p^*-1|^{p^*-1}K^{(1)}_{1-\lambda,\frac 1{1-p}}$. Throughout the paper we adopt the following notation
\begin{equation}
	\widetilde{g}_{p,\lambda}=\begin{cases}
g_{p,\lambda},\quad\lambda>0,\\
\overline g_{p,\lambda},\quad \lambda<0.
	\end{cases}
\end{equation}
\medskip

\noindent \textbf{The Cram\'er-Rao inequality} is an informational inequality relating the moments $\sigma_p$ and the $(p,\lambda)$-Fisher information and which follows trivially from \eqref{ineq:bip_E-M} and \eqref{ineq:bip_Stam} by multiplication. However, we state it below since it will be employed in the paper. When $p^*\in[0,\infty )$ and $\lambda>\frac1{1+p^*}$, we have
\begin{equation}\label{ineq:CR}
\phi_{p,\lambda}[\pdf]\sigma_{p^*}[\pdf]\geqslant \phi_{p,\lambda}[\gauss_{p,\lambda}]\sigma_{p^*}[\gauss_{p,\lambda}]\equiv \widehat{\kappa}^{(0)}_{p,\lambda}:= K^{(0)}_{p,\lambda}K^{(1)}_{p,\lambda},
\end{equation}
for any absolutely continuous probability density function.

\medskip

\noindent \textbf{The extended (tri-parametric) Stam inequality}. This is a generalized inequality, valid both in the classical domain of parameters and in a mirrored domain, established by the authors in the previous work \cite[Theorem 5.1]{IP25} with the help of the up and down transformations that will be recalled below. However, since the statement is completely independent of them, we give its statement here, in order to complete the list of informational inequalities that will be used throughout the paper. Let $p\geqslant 1$ and $\beta$ be such that
\begin{equation}\label{eq:sign_cond1}
\sign\left(p^*\beta + \lambda-1\right)=\sign\left(\beta+1-\lambda\right)\neq0.
\end{equation}
Then, the following generalized Stam inequality holds true for $f: \Rset\mapsto\Rset^+$ absolutely continuous if $1+\beta-\lambda>0$ or for $f:(x_i,x_f)\mapsto \Rset^+$ absolutely continuous on $(x_i,x_f)$ if $1+\beta-\lambda<0$:
\begin{equation}\label{ineq:trip_Stam_extended}
	\left(\phi_{p,\beta}[\pdf] \, N_{\lambda}[\pdf] \right)^{\theta(\beta,\lambda)}\: \geqslant \:	\left(\phi_{p,\beta}[\widetilde g_{p,\beta,\lambda}] \, N_{\lambda}[\widetilde g_{p,\beta,\lambda}] \right)^{\theta(\beta,\lambda)}\equiv K_{p,\beta,\lambda}^{(1)},
\end{equation}
where $\theta(\beta,\lambda)=1+\beta-\lambda$ and $\widetilde g_{p,\beta,\lambda}$ is defined in \cite[Section 5]{IP25}. Moreover, for $p<1$ and $\beta, \lambda$ such that
\begin{equation}\label{eq:sign_cond2}
\sign(p^*\beta+\lambda-1)=\sign(\beta+\lambda-1)\neq0, \quad \sign(\lambda-1)=\sign(\beta)\neq0.
\end{equation}
the inequality ~\eqref{ineq:trip_Stam_extended} holds true with $\theta(\beta,\lambda)=-\beta$, provided that $f:\Rset \mapsto\Rset^+$ is continuously differentiable with $f'<0$. The optimal constants and the minimizers $\widetilde g_{p,\beta,\lambda}$ are known, but we omit the (rather technical) expressions here and we refer the interested reader to \cite[Section 5]{IP25} for them.

\subsection{Up and down transformations}\label{sec:updown}

This section is devoted to a brief recall of the recently introduced up/down transformations~\cite{IP25}. They motivate the definitions of the informational functionals investigated throughout the paper and are essential tools in the proofs of the informational inequalities in Section \ref{sec:ineq}. Let us mention here that, throughout the paper, $\Omega=(x_i,x_f)$ represents either a bounded or an unbounded interval such that $x_i>-\infty$, while $x_f$ is either finite or equal to $+\infty$. As a convention, the integrals will be generally written over $\Rset$, but understanding that the integrals are actually taken over the support of $f$.` 

\medskip

\noindent \textbf{The down transformation.} To begin with, \emph{the down transformation} establishes a bijection between the set of decreasing density functions and a more general class of density functions. The definition mixes up the proper probability density function $\pdf$ and its derivative and also involves a change of the independent variable, as follows:
\begin{definition}\label{def:down}
Let $\pdf: \supp\longrightarrow \Rset^+$ be a probability density function with $\supp=(x_i,x_f),$ where $-\infty<x_i<x_f\le\infty$, such that $\pdf'(x)<0,\;\forall x\in\supp$. Then, for $\alpha\in\Rset$, we define the transformation $\down{\alpha}[\pdf(x)]$ by
	\begin{equation}\label{eq:down}
	\down{\alpha}[\pdf(x)](s)=\pdf^\alpha(x(s))|\pdf'(x(s))|^{-1},\qquad s'(x)=\pdf^{1-\alpha}(x) |\pdf'(x)|.
	\end{equation}
\end{definition}
It is easy to see that $\down{\alpha}[\pdf]$ is a probability density function. For the sake of simplicity, we denote the down transformed density by
$$
\pdf^{\downarrow}_\alpha\equiv\down{\alpha}[\pdf].
$$	
Let us also observe that the class of transformations $\down{\alpha}$ is defined up to a translation, since $s(x)$ depends on an additive integration constant. Without loss of generality, we consider the \emph{canonical election}
\begin{equation}\label{eq:can_change}
s(x)=\left\{\begin{array}{ll}\frac{\pdf(x)^{2-\alpha}}{\alpha-2}, & {\rm for} \ \alpha\in\Rset\setminus\{2\},\\
-\ln\,\pdf(x), & {\rm for} \ \alpha=2,\end{array}\right.
\end{equation}	
With this canonical election, it has been noticed that the quantity $(\alpha-2)s$ is always positive. With respect to the length of the support, when $\alpha<2$ the support $\supp^\downarrow_\alpha$ is bounded whenever $\pdf$ is a bounded density, and unbounded for probability densities such that $\pdf(x)\rightarrow\infty$ as $x\to x_i$. On the contrary, for $\alpha\geqslant 2$ the support turns out to be bounded when $\pdf(x_f)>0$ and unbounded when $\pdf(x_f)=0.$ Moreover, if $\pdf'(x)\rightarrow 0$ when $x\rightarrow x_i$, then $\adown(s)\rightarrow \infty$ when $s\rightarrow s_i$, thus we can obtain densities that are divergent at the border of their support as a result of applying the down transformation. The next remark is very useful in the sequel.

\begin{remark}[Derivatives and composition]\label{rem:downcomposition}
In order to apply the down transformation twice, we need the function obtained after the first transformation to be monotone. We have
\begin{equation}\label{eq:der_down}
\begin{split}
	\frac{d {\adown}}{ds}=\frac{\pdf^{2\alpha-2}(x)}{f'(x)}\left(\alpha  -\frac{\pdf(x)\pdf''(x)}{[\pdf'(x)]^{2}}\right),
\end{split}
\end{equation}
from where the sign of the derivative of $\adown$ depends exclusively of ${\rm sign}\left(\alpha -\frac{\pdf(x)\pdf''(x)}{[\pdf'(x)]^{2}}\right).$
\end{remark}

\medskip

\noindent \textbf{The up transformation.} \emph{The up transformation} is defined through the inverse function of a kind of \textit{incomplete expected value} of the independent variable. Contrary to the down transformation, the up transformation is applicable to any probability density function.
\begin{definition}\label{def:up}
Let $\pdf: \supp\longrightarrow \mathbb R^+$ be a probability density function with $\supp=(x_i,x_f)$. For $\alpha\in\Rset\setminus\{2\}$ we introduce the up transformation $\up{\alpha}$ by
\begin{equation}\label{eq:up}
	f^{\uparrow}_\alpha(u)=\up{\alpha}[\pdf(x)](u)=|(\alpha-2)x(u)|^\frac{1}{2-\alpha},\quad u(x)=\int_{x}^{x_f} |(\alpha-2)v|^\frac{1}{\alpha-2}f(v)dv,
\end{equation}
while for $\alpha=2$ we set
\begin{equation}\label{eq:up2}
f^{\uparrow}_2(u)=\up{2}[\pdf(x)](u)=e^{-x(u)},\quad u(x)=\int_{x}^{x_f} e^vf(v)dv,
\end{equation}
\end{definition}
We precise here that the definition of $u(x)$ in Eqs. \eqref{eq:up} and~\eqref{eq:up2} is taken up to a translation; that is, for simplicity, in the sequel we use the primitive
$$
u(x)=\int_x^{x_f}|(\alpha-2)x|^\frac{1}{\alpha-2}f(x)dx, \quad \alpha\neq2,
$$
and the similar one for $\alpha=2$, where $x_f\in\Rset\cup\{\infty\}$ is the upper edge of the support of the domain of $f$, whenever this integral is finite. In the contrary case, we can employ any intermediate point $x_0$ in $(x_i,x_f)$ if the integral is divergent at its end.

According to the definitions, it should not be a surprise that the up transformation is in fact is the inverse of the down transformation.
\begin{proposition}\label{prop:inv}
Let $\pdf$ a probability density and let $\adown=\down{\alpha}[f]$ and $\aup=\up{\alpha}[f]$ be its $\alpha$-order down and up transformations. Then, up to a translation,
	\begin{equation}
	\pdf=\up{\alpha}[\adown]=\down{\alpha}[\aup].
	\end{equation}
that is, $\down{\alpha}\up{\alpha}=\up{\alpha}\down{\alpha}=\mathbb I$, where $\mathbb I$ denotes the identity operator (up to a translation).
\end{proposition}	
Let us observe that the support of the up transformed density is limited by $u_i=u(x_i)=\int_{x_i}^{x_f} |(\alpha-2)v|^\frac1{\alpha-2}f(v)dv$ and $u_f=u(x_f)=0$, thus it is the interval $(0,u_i)$, where $u_i$ can be finite or infinite.

\medskip

\noindent \textbf{Behavior with respect to scaling changes.} We gather in this paragraph the results related to the behavior of the down and up transformations with respect to a rescaling of the probability density function $f(x)$. For any $\kappa\in\Rset^+$, let us define the rescaled density
\begin{equation}\label{def:resc}
f^{[\kappa]}(x)=\kappa f(\kappa x).
\end{equation}

\begin{proposition}[Scaling and up/down transformations]\label{prop:scaling_down}
Let $\alpha\in\Rset\setminus\{2\}$ and $\kappa\in\Rset^+$. Then, for any probability density $f$ satisfying the requirements of Definition \ref{def:down}, we have
\begin{equation}\label{eq:down_resc}
	\down{\alpha}[f^{[\kappa]}](s)=\left(\down{\alpha}[f](s)\right)^{[\kappa^{\alpha-2}]}.
\end{equation}
and
\begin{equation}\label{eq:scaling_up}
	\up{\alpha}[f^{[\kappa]}](u)=\left(\up{\alpha}[f](u)\right)^{[\kappa^{1/(\alpha-2)}]},
\end{equation}
When $\alpha=2$, we have
\begin{equation}\label{eq:down_resc_a2}
	\down{2}[f^{[\kappa]}(s)]=\down{2}[f](s+\ln \kappa).
\end{equation}
\end{proposition}
The behavior of the up transformation with $\alpha=2$ with respect to the scaling change \eqref{def:resc} is not so straightforward and also not needed in the present work. We refer the interested reader to our previous paper \cite{IP25}.

\medskip

\noindent \textbf{Moments and entropies of up/down transformed densities.} We end up this review of the properties of the up and down transformations needed in the present work by the following technical but important result, dealing with the expression of $p$-typical deviations, the R\'enyi and Shannon entropy power and the Fisher information of up and down transformed densities.
\begin{lemma}\label{lem:MEF}
Let $\pdf$ be a probability density and $\aup$ and $\adown$ its up/down transformations. Then, if $\alpha\in\Rset\setminus\{2\}$, the following equalities hold true:
	\begin{equation}\label{eq:ME}
	\sigma_p[\adown]=\frac{N_{1+(2-\alpha)p}^{\alpha-2}[\pdf]}{|2-\alpha|},\quad\text{or equivalently,}\quad N_{\lambda}[\aup]=\left(|2-\alpha| \sigma_{\frac{\lambda-1}{2-\alpha}}[\pdf]\right)^{\frac{1}{\alpha-2}},
	\end{equation}
and
	\begin{equation}\label{eq:EF}
	N_{\lambda}[\adown]=\phi_{1-\lambda,2-\alpha}^{2-\alpha}[\pdf],\quad\text{or equivalently,}\quad  \phi_{p,\beta}[f^{\uparrow}_{2-\beta}]=\left(N_{1-p}[\pdf]\right)^{\frac1{\beta}}.
	\end{equation}
For the Shannon entropy we have
	\begin{equation}\label{eq:SUD}
	S[\adown]=\alpha S[f]+\big \langle \log \left| f'\right|\big \rangle,\qquad S[\aup]=\frac1{2-\alpha}\left\langle \log |x|\right \rangle+\frac{\log|2-\alpha|}{2-\alpha}.
	\end{equation}
For $\alpha=2$, we have the following equalities:
\begin{equation}\label{eq:Sp_down2}
\sigma_p[\down{2}[\pdf]]=\left[\int_{\Rset}f(x)|\ln\,f(x)|^pdx\right]^{\frac{1}{p}},
\end{equation}
respectively
\begin{equation}\label{eq:Ren_down2}
N_\lambda[\down{2}[\pdf]]=\lim_{\widetilde \lambda\to0}\phi_{1-\lambda,\widetilde \lambda}[f]^{\widetilde\lambda}\equiv F_{1-\lambda,0}[f]^\frac{1}{1-\lambda}.
\end{equation}
\end{lemma}
The proof of this Lemma can be found in \cite[Section 3.1]{IP25}. We end this section by establishing a new preparatory result regarding the Shannon entropy, which is a direct consequence of Eqs.~\eqref{eq:der_down} and ~\eqref{eq:SUD}. Introducing the following notation	
\begin{equation*}
f^{\downarrow\downarrow}_{\alpha\beta}=\down{\beta}\left[\down{\alpha}[f]\right],
\end{equation*}
we have
\begin{proposition}
	Let $f$ be a strictly decreasing probability density and let $\alpha,\beta$ be two real numbers such that
	$$
	\sup\limits_{x\in\Rset}\left(\frac{\pdf(x)\pdf''(x)}{[\pdf'(x)]^{2}}\right)<\alpha.
	$$
	Then
	\begin{equation}\label{eqprop:downdown}
	S[f^{\downarrow\downarrow}_{\alpha\beta}]
	=(\alpha\beta-2\alpha+2) S[f]
	+(\beta-1)\left\langle\log|f'|\right\rangle
	+\left\langle\log\left(\alpha-\frac{f f''}{(f')^2}\right)\right\rangle
	\end{equation}
\end{proposition}
\begin{proof}
	For any $\beta\in\Rset$, we infer from the left equation in Eq.~\eqref{eq:SUD} and the expresion of the derivative of the down transformed density given in Eq.~\eqref{eq:der_down} that
	\begin{equation}
	\begin{split}
	S[\down\beta[\down\alpha [f]]]&=S[\down\beta[\adown]]=\beta S[\adown]+\left \langle \log \left|\left(\adown\right)' \right|\right\rangle
	\\
	&=
	\beta S[\adown]+(2\alpha-2)\langle\log f \rangle -\langle\log |f'| \rangle +\left\langle\log\left(\alpha-\frac{f f''}{(f')^2}\right)\right\rangle
	\\
	&=
	\beta \left(\alpha S[f]+\left\langle \log|f'|\right\rangle\right)-(2\alpha-2)S[f]-\langle\log |f'| \rangle +\left\langle\log\left(\alpha-\frac{f f''}{(f')^2}\right)\right\rangle,
	\end{split}
	\end{equation}
	leading to Eq.~\eqref{eqprop:downdown} after regrouping terms.
\end{proof}

We are now ready to introduce the new informational functionals announced in the Introduction.

\section{Upper-moments and down-Fisher measures: basic properties}\label{sec:newinf}

In this section we introduce our main objects of study in the present work. These are new informational functionals that are obtained, on the one hand, by computing the $p$-moments of probability density functions obtained by a finite number of applications of the up transformation and, on the other hand, by computing the Fisher information of a down transformed density.

\subsection{Upper-moments: the first order case}

In a first step, we apply a single up transformation and we introduce the following upper-moments.

\begin{definition}\label{def:hypermom}[upper-moments]
Let $p\in\Rset$ and $f:\supp\rightarrow \Rset^+$ be a probability density. Then, for $\alpha\neq2$, we define the $(p,\alpha)-$\textit{upper-moments}, $ M_{p,\alpha}[\pdf],$ as
\begin{equation}\label{eq:hyper}
 M_{p,\alpha}[\pdf]=\int_{\Rset}\left|\int_x^{x_f} |(\alpha-2) v|^\frac1{\alpha-2}f(v)dv \right|^p\pdf(x)dx.
\end{equation}
For $\alpha=2$, the $(p,2)-$\textit{upper-moments} $ M_{p,2}[\pdf]$ is defined as
\begin{equation}\label{eq:hyper2}
M_{p,2}[\pdf]=\int_{\Rset}\left| \int_x^{x_f} e^vf(v)dv \right|^p\pdf(x)dx.
\end{equation}
Finally, for $\alpha\neq2$ we define the $(p,\alpha)-$\textit{upper-deviation} as
\begin{equation}\label{eq:hyperdev}
m_{p,\alpha}[f]=M_{p,\alpha}[f]^\frac{\alpha-2}p.
\end{equation}
\end{definition}
The next remark gives the idea of how we actually arrived at the definition of the upper-moments.
\begin{remark}
Let $\alpha$, $p\in\Rset$ and let $f$ be a probability density whose up transformed density is denoted by $\aup$. Then, an immediate application of Definition \ref{def:up} of the up transformation and Definition \ref{def:hypermom} of the upper-moments leads to
\begin{equation}\label{eq:mom_up}
\mu_{p}[\aup]=M_{p,\alpha}[\pdf].
\end{equation}
Moreover, observe that the upper-deviation $m_{p,\alpha}$ inherits the order relation with respect to the parameter $p$ from the $p$-deviations when $\alpha>2$; that is, for any $p<q\in\Rset$ and $\alpha>2$, we have
$$
m_{p,\alpha}[f]=\mu_p[\aup]^{\frac{\alpha-2}{p}}\leqslant \mu_{q}[\aup]^{\frac{\alpha-2}{q}}=m_{q,\alpha}[f],
$$
while the inequality sign is reversed if $\alpha<2$.
\end{remark}
We examine the behavior of the upper-deviation with respect to a scaling change in the next result.
\begin{proposition}[Scaling changes] Given $\alpha\in\Rset\setminus\{2\}$ and $p\in\Rset$, for any $\kappa>0$ and for any probability density $f$ we have
\begin{equation}\label{eq:hyper_scaling}
m_{p,\alpha}[f^{[\kappa]}]=\kappa^{-1}m_{p,\alpha}[f].
\end{equation}
\end{proposition}
\begin{proof}
Since
$$
\sigma_p[f^{[\kappa]}]=\kappa^{-1}\sigma_p[f],
$$
we readily infer from Proposition~\ref{prop:scaling_down} and Eq.~\eqref{eq:mom_up} that
\begin{equation*}
m_{p,\alpha}[f^{[\kappa]}]=\sigma_{p}[\up{\alpha}[f^{[\kappa]}]]^{\alpha-2}=\sigma_{p}[(\up{\alpha}[f])^{[\kappa^{1/(\alpha-2)}]}]^{\alpha-2}
=\kappa^{-1}\sigma_{p}[\up{\alpha}[f]]^{\alpha-2}=\kappa^{-1}m_{p,\alpha}[f],
\end{equation*}
completing the proof.
\end{proof}

\subsection{Upper-moments: higher order cases}

We are now ready to generalize the previous construction to upper-moments of any finite order. In order to fix the notation, let $\vec\alpha=(\alpha_0,\alpha_1,...,\alpha_{n-1})\in\Rset^n$ be any $n$-dimensional vector. In a first step, for the reader's convenience, we give a complete definition for second order upper-moments.
\begin{definition}[Second order $(\vec\alpha,p)$-upper-moments]\label{def:hyper2}
Let $n=2$ and $\vec\alpha=(\alpha_0,\alpha_1)$ with $\alpha_i\neq2$, $i=1,2$. We then define
\begin{equation*}
M_{p,\vec\alpha}[f]=K(p,\vec\alpha)\int_\Rset
\left|\int_{x_0}^{x_{0f}}\left|\int_{x_1}^{x_{1f}} |t|^\frac{1}{\alpha_1-2} f(t)dt\right|^\frac{1}{\alpha_0-2}f(x_1)dx_1\right|^pf(x_0) dx_0,
\end{equation*}
where
		\begin{equation*}
K(p,\vec \alpha)=|\alpha_0-2|^{\frac{p}{\alpha_0-2}}|\alpha_1-2|^{\frac{p}{(\alpha_0-2)(\alpha_1-2)}}
		\end{equation*}
If $\alpha_0\neq2$ and $\alpha_1=2$, then the definition is adapted as follows:
\begin{equation*}
M_{p,\vec\alpha}[f]=|\alpha_0-2|^{\frac{p}{\alpha_0-2}}\int_\Rset
\left|\int_{x_0}^{x_{0f}}\left|\int_{x_1}^{x_{1f}} e^t f(t)dt\right|^\frac{1}{\alpha_0-2}f(x_1)dx_1\right|^pf(x_0) dx_0,
\end{equation*}
while if $\alpha_0=2$ and $\alpha_1\neq2$, we write
\begin{equation*}
M_{p,\vec\alpha}[f]=\int_\Rset\left|\int_{x_0}^{x_{0f}}\exp\left\{\left|\int_{x_1}^{x_{1f}}|(\alpha_1-2)t|^\frac{1}{\alpha_1-2}f(t)dt\right|\right\}f(x_1)dx_1\right|^pf(x_0) dx_0.
\end{equation*}
We omit here the definition for $\vec\alpha=(2,2)$, which is rather tedious and not of much interest.
\end{definition}
We generalize the previous definition, following what is a rather obvious rule of extension, to upper-moments of order $n$. For simplicity, we give this definition only when assuming that all the components of the vector $\vec\alpha$ are different from two, the adaptation of the definition if, for example, $\alpha_{n-1}=2$ being straightforward and similar to the one we explicitly wrote in Definition \ref{def:hyper2}.
\begin{definition}\label{def:hypergen}
Let now $n\geqslant1$ be an integer and $\vec\alpha=(\alpha_0,\alpha_1,...,\alpha_{n-1})\in(\Rset\setminus\{2\})^n$. We then define		
\begin{equation}\label{eq:hypern}
\begin{split}
M_{p,\vec \alpha}[f]&=K(p,\vec\alpha)\int_\Rset
\left|\int_{x_0}^{x_{0f}}\left|\int_{x_1}^{x_{1f}} \ldots \Bigg|\int_{x_{n-1}}^{x_{n-1,f}} |t|^\frac{1}{\alpha_{n-1}-2} f(t)dt\right|^\frac{1}{\alpha_{n-2}-2}\times \right. \\
&\left.\left.\times\;f(x_{n-1})dx_{n-1} \Bigg |^\frac{1}{\alpha_{n-3}-2} f(x_{n-2})dx_{n-2}\ldots\right|^{\frac1{\alpha_0-2}}f(x_1)dx_1\right|^pf(x_0) dx_0,
\end{split}
\end{equation}
where
\begin{equation*}
K(p,\vec \alpha)=|\alpha_0-2|^{\frac{p}{\alpha_0-2}}|\alpha_1-2|^{\frac{p}{(\alpha_0-2)(\alpha_1-2)}}\ldots |\alpha_{n-1}-2|^{p\prod\limits_{i=0}^{n-1}\frac1{\alpha_i-2}}.
\end{equation*}
Similarly as in Definition \ref{def:hypermom}, we define the $(p,\vec\alpha)-$\textit{upper-deviation} as
\begin{equation}\label{eq:hyperdevgen}
m_{p,\vec\alpha}[f]=M_{p,\vec\alpha}[f]^{\frac{(\alpha_0-2)(\alpha_1-2)\ldots(\alpha_{n-1}-2)}{p}}.
\end{equation}		
\end{definition}
An easy induction argument, together with the previous definition and Definition \ref{def:up}, entail that, for any vector of length $n$, $\vec\alpha=(\alpha_0,\alpha_1,...,\alpha_{n-1})\in(\Rset\setminus\{2\})^n$, we have
\begin{equation}\label{eq:mom_up(n)}
M_{p,\vec \alpha}[f]=\mu_p[\up{\alpha_0}[\up{\alpha_1}[...\up{\alpha_{n-1}}[f]]]].
\end{equation}
We also denote with $\widetilde{M}_{p,\vec \alpha}[f]$ the quantity obtained in a similar way as $M_{p,\vec \alpha}[f]$ but removing the external absolute value and letting only $p\in\mathbb{N}$; that is, in the definition of $M_{p,\vec \alpha}[f]$, we replace $|\cdot|^p$ by $(\cdot)^p$. This notation will be employed only in Section \ref{sec:Hausdorff}.

\subsection{Down-Fisher measures: basic properties}\label{sec:DF}
In this section we go in the opposite direction with respect to the previous ones: while in the previous sections we have introduced the upper-moments employing the up transformation, in the current one we employ the down transformation to define \emph{down-Fisher measures} as indicated in the next definition.

\begin{definition}[Down-Fisher measures]\label{def:sub-fisher}
Let $f$ be a probability density function and $p$, $q$, $\lambda$ three real numbers such that $p\neq q$. We introduce the \emph{down-Fisher measure}
\begin{equation}\label{eq:sub-fisher}
\varphi_{p,q,\lambda}[\pdf]=\int_\supp f(v)^{1+p(\lambda-2)}|f(v)'|^{q}\left|\frac{p\lambda}{p-q}-\frac{\pdf(v) \pdf(v)''}{(\pdf(v)')^2}\right|^p\,dv.
\end{equation}	
\end{definition}
Similarly as in the definition of the upper-moments, the following property illustrates how we actually arrived to the expression given in Eq. \eqref{eq:sub-fisher}.
\begin{lemma}\label{lem:sub-fisher}
Let $f$, $p$, $q$ and $\lambda$ as in Definition \ref{def:sub-fisher} and $\alpha\in\Rset$. Then we have
\begin{equation}\label{eq:fpl}
F_{p,\lambda}[\adown]=\varphi_{p,q,\overline{\lambda}}[\pdf],\quad q=p(1-\lambda),\quad \overline{\lambda}=\alpha\lambda.
\end{equation}
\end{lemma}
\begin{proof}
We calculate the left hand side of Eq. \eqref{eq:fpl} as follows:
\begin{equation*}
\begin{split}
F_{p,\lambda}[\adown]&=\int_{\Rset}\adown(s)^{1+(\lambda-2)p}\left|\frac{d\adown}{ds}(s)\right|^p\,ds\\
&=\int_{\Rset}\left[\frac{f^{\alpha}(x)}{|f'(x)|}\right]^{(\lambda-2)p}
\left|\frac{f^{2\alpha-2}(x)}{f'(x)}\left[\alpha-\frac{f(x)f''(x)}{(f'(x))^2}\right]\right|^pf(x)\,dx\\
&=\int_{\Rset}f(x)^{1+p(\alpha\lambda-2)}|f'(x)|^{p(1-\lambda)}\left|\alpha-\frac{f(x)f''(x)}{(f'(x))^2}\right|^p\,dx,
\end{split}
\end{equation*}
where in the first equality step we have employed Eq. \eqref{eq:der_down}. Letting now
\begin{equation}\label{eq:interm3}
q=p(1-\lambda), \quad \overline{\lambda}=\alpha\lambda,
\end{equation}
we directly arrive to Eq. \eqref{eq:fpl}, noticing that the equalities \eqref{eq:interm3} readily give $\alpha=p\overline{\lambda}/(p-q)$.
\end{proof}

In a stark contrast with respect to the up transformation, which can be applied as many times as we wish to any density function in order to define (based on it) upper-moments of any order, a successive application of the down transformation is generally troublesome because of \eqref{eq:der_down} and the fact that the down transformation only applies to decreasing functions. This is why, we refrain from defining lower order down-Fisher measures, as the definition would lack generality. However, there are at least two categories of density functions for which we can apply the down transformation infinitely many times:

\smallskip

$\bullet$ \emph{exponential densities}: assume that $f:(R,\infty)\mapsto\Rset$, $f(x)=Ae^{A(R-x)}$, for positive constants $A$ and $R$. It is immediate to check that $f$ is a probability density function. We compute its down transformation using its definition. Letting $C=Ae^{AR}$ in order to simplify the notation, we have for any $\alpha\in\Rset\setminus\{2\}$
$$
s(x)=\frac{f(x)^{2-\alpha}}{\alpha-2}=\frac{C^{2-\alpha}}{\alpha-2}e^{(\alpha-2)Ax},
$$
or equivalently
$$
x(s)=\frac{1}{(\alpha-2)A}\ln\frac{(\alpha-2)s}{C^{2-\alpha}}.
$$
Thus, the down transformed density of $f$ (with $\alpha\neq2$) becomes
$$
\adown(s)=\frac{C^{\alpha-1}}{A}e^{-(\alpha-1)Ax(s)}=\frac{C^{\alpha-1}}{A}
\left(\frac{\alpha-2}{C^{2-\alpha}}\right)^{-\frac{\alpha-1}{\alpha-2}}s^{-\frac{\alpha-1}{\alpha-2}}
=\frac{1}{A(\alpha-2)^{\frac{\alpha-1}{\alpha-2}}}s^{-\frac{\alpha-1}{\alpha-2}},
$$
which remains decreasing, provided $\alpha>2$ or $\alpha<1$. We observe that, at least for $\alpha>2$, the down transformation mapped the exponential density into a power density at infinity.

$\bullet$ \emph{power densities}: assume that $f:(R,\infty)\mapsto\Rset$, $f(x)=Cx^{-\eta}$ for $\eta>1$ and some positive normalization constant $C$. Then, for $\alpha\neq2$, 
$$
s(x)=\frac{f(x)^{2-\alpha}}{\alpha-2}=\frac{C^{2-\alpha}}{\alpha-2}x^{\eta(\alpha-2)}
$$
or, equivalently,
$$
x(s)=\left[\frac{(\alpha-2)s}{C^{2-\alpha}}\right]^{\frac{1}{\eta(\alpha-2)}}.
$$
We compute now the down transformed density 
\begin{equation}\label{eq:interm2}
\begin{split}
\adown(s)&=\frac{C^{\alpha-1}x(s)^{-\alpha\eta}}{\eta x(s)^{-(\eta+1)}}=\frac{C^{\alpha-1}}{\eta}x(s)^{1+\eta-\alpha\eta}
=\frac{C^{\alpha-1}}{\eta}\left[\frac{(\alpha-2)s}{C^{2-\alpha}}\right]^{\frac{1+\eta-\alpha\eta}{\eta(\alpha-2)}}\\
&=\frac{C^{\frac{1}{\eta}}}{\eta}\big[(\alpha-2)s\big]^{\frac{1+\eta-\alpha\eta}{\eta(\alpha-2)}}.
\end{split}
\end{equation}
Recalling that $\eta>1$, we observe that, if we pick $\alpha>2$, we have
$$
\frac{\alpha\eta-1-\eta}{\eta(\alpha-2)}=\frac{\alpha-1-1/\eta}{\alpha-2}>1,
$$
hence \eqref{eq:interm2} ensures that we can keep applying the down transformation one more step and the probability density function we obtain will still have an integrable tail as $s\to\infty$, which allows us to continue applying the down transformation as before. In the case $\alpha<\frac{\eta+1}{\eta}$ we also obtain a decreasing density, but in this case with a compact support. This ensures that we can apply one more time the down transformation as well, but we do not see a repeating structure in order to generalize it infinitely many times in this case.

We thus deduce that, for both decreasing power density functions and decreasing exponential density functions, we can apply the down transformation with exponents $\alpha>2$ infinitely many times. We conjecture that these functions are the only ones allowing for an infinite number of applications of the down transformation.

\subsection{Order relation between down-Fisher measures}
We next establish an ordering for down-Fisher measures with different parameters. As a limiting case of it, we also derive an informational inequality involving down-Fisher measures and the Shannon entropy power.
\begin{theorem}\label{th:down-fisher}
	Let $p$, $q\in\Rset\setminus\{0\}$ such that $p> q$. Then, for any $\lambda\in\Rset$, $r\in\Rset\setminus\{p\}$ and $f$ a probability density function, we have the following ordering between down-Fisher measures:
	\begin{equation}\label{eq:down-order}
	\varphi_{p,r,\lambda}^{\frac{1}{p}}[f]\geqslant\varphi_{q,qr/p,\lambda}^{\frac{1}{q}}[f].
	\end{equation}
Moreover, the following inequalities involving the Shannon entropy, which can be seen as limiting cases of Eq. \eqref{eq:down-order},
	\begin{equation}\label{eq:Shannon-down}
	\varphi_{q,0,\alpha}^{\frac{1}{q}}[f]e^{(\alpha-2)S[f]}\left\{\begin{array}{ll}
	\geqslant\exp\left\{\left\langle\log\left(\alpha-\frac{ff''}{(f')^2}\right)\right\rangle\right\}, & q>0,\\[1mm]
	\leqslant\exp\left\{\left\langle\log\left(\alpha-\frac{ff''}{(f')^2}\right)\right\rangle\right\}, & q<0,\end{array}\right.
	\end{equation}
	hold true for any probability density function $f$ and for any $\alpha\in\Rset\setminus\{2\}$. The equality is achieved when
\begin{equation}\label{eq:minim}
f_{\rm min}=\up{\frac{r+p}p}\up{\frac{\lambda p}{p-r}}[u],
\end{equation}
where $u$ denotes the uniform density $u(x)=\frac{1}{x_f-x_i},\;x\in[x_i,x_f].$
\end{theorem}
One can verify that this family of minimizing densities contains a rich diversity of probability density functions, which may be expressed through special functions, such as generalized hyperbolic sine functions, the inverse of the incomplete Gamma, the inverse of the Gauss hypergeometric function $_2F_1$, or the inverse of the exponential integral $Ei$. Due to the technicality of these special functions, we avoid entering a deeper discussion of them.
\begin{proof}
	Let us first observe that, taking into account the rather complex expression of the down-Fisher measure, a direct, analytical proof of Eq. \eqref{eq:down-order} seems far from obvious. We instead start from the following easy inequality (see for example \cite{Beck93})
	\begin{equation}\label{ineq:easy}
	N_{\lambda}[f]\geqslant N_{\beta}[f], \quad {\rm if} \ \lambda<\beta,
	\end{equation}
	that we apply to a down transformed density $\adown$ with $\alpha\neq2$. Let us first consider $\alpha<2.$ By applying Lemma \ref{lem:MEF}, we obtain
	$$
	1\leqslant\frac{N_{\lambda}[\adown]}{N_{\beta}[\adown]}=\frac{\phi_{1-\lambda,2-\alpha}^{2-\alpha}[f]}{\phi_{1-\beta,2-\alpha}^{2-\alpha}[f]},
	$$
	which entails, after setting $p=1-\lambda$, $q=1-\beta$ and $\gamma=2-\alpha>0$, that
	\begin{equation}\label{eq:interm7}
	\phi_{p,\gamma}[f]\geqslant\phi_{q,\gamma}[f], \quad {\rm for \ any} \ p\geqslant q, \ \gamma>0.
	\end{equation}
Note that Eq. \eqref{eq:interm7} can be also obtained as a consequence of H\"older's inequality. We apply the inequality \eqref{eq:interm7} to another down transformed density function $\xidown$ for some $\xi\in\Rset$ and we deduce from Lemma \ref{lem:sub-fisher} that
	\begin{equation}\label{eq:interm7bis}
	\varphi_{p,p(1-\gamma),\xi\gamma}^{\frac{1}{p\gamma}}[f]=\phi_{p,\gamma}[\xidown]\geqslant\phi_{q,\gamma}[\xidown]
	=\varphi_{q,q(1-\gamma),\xi\gamma}^{\frac{1}{q\gamma}}[f],
\end{equation}
for $p>q$ and $\gamma>0$. We next introduce the notation $r=p(1-\gamma)$, $\lambda=\xi\gamma$ and, noticing that $q(1-\gamma)=qr/p$, we readily derive the inequality \eqref{eq:down-order}. We observe that if we let $\alpha>2$,  the inequality sign in Eq.~\eqref{eq:interm7} and Eq. \eqref{eq:interm7bis} is reversed but, taking into account that $\gamma=2-\alpha<0$ in this case, the inequality \eqref{eq:down-order} follows once more from Eq. \eqref{eq:interm7bis} by raising everything to power $\gamma$. Since the uniform density $u=\frac{1}{x_f-x_i}$ on the bounded interval $[x_i,x_f]$ is the minimizer of the inequality~\eqref{ineq:easy}, then the minimizing densities of the inequality ~\eqref{eq:interm7bis} must satisfy $\down{\xi}\down{\alpha}[f_{\rm min}]=u.$ By applying the inversion property in Proposition~\ref{prop:inv}, we obtain the claimed expression \eqref{eq:minim} after recalling that
$$
\alpha=2-\gamma=\frac{p+r}{p}, \quad \xi=\frac{\lambda}{\gamma}=\frac{p\lambda}{p-r}.
$$
	
In order to prove the inequality \eqref{eq:Shannon-down}, we start again from the inequality \eqref{ineq:easy}, which we apply for $\lambda=1$ and, for simplicity, $\beta<1$ (the opposite case being completely similar and involving only a change of the inequality sign in the forthcoming steps). We thus have
$$
e^{S[f]}\leqslant N_{\beta}[f], \quad \beta<1,
$$
which we apply to the density $f^{\downarrow\downarrow}_{\alpha1}=\down{1}[\down{\alpha}[f]]$. We infer from Eqs. \eqref{eqprop:downdown} and \eqref{eq:fpl} that
$$
	e^{(2-\alpha)S[f]}\exp\left\{\left\langle\log\left(\alpha-\frac{ff''}{(f')^2}\right)\right\rangle\right\}
	\leqslant\varphi_{1-\beta,0,\alpha}^{\frac{1}{1-\beta}}[f],
	$$
	which is equivalent to \eqref{eq:Shannon-down} by setting $q=1-\beta>0$. In the case $\beta>1$, recalling the convention related to Fisher-type measures with negative parameters as explained in the Introduction, we follow the same steps as above and only change the inequality sign from $\leqslant$ to $\geqslant$.
\end{proof}

\subsection{Dealing with the moment problems}\label{sec:Hausdorff}

In this short section, we give some applications of the upper-moments introduced in the previous parts of this work to moment problems related to probability density functions. Let us first recall that, in our previous work \cite{IP25}, we gave an equivalence of the Hausdorff moment problem with an entropy problem, as well as with a Fisher counterpart, and we extended the MaxEnt approach to a more general framework involving the Rényi entropy power as the fixed functional and the Fisher information as the functional which is maximized. We further extend below this approach to upper-moments of any order. With respect to the Hausdorff moment problem and its equivalent entropy moment problem, we have the following immediate consequence of the definition of the upper-moments, Eq. \eqref{eq:mom_up(n)} and the inversion property Proposition \ref{prop:inv} of the up and down transformations.
\begin{corollary}\label{cor:MP}
Let $\{m_i\}_{i=1}^{\infty}$ be a set of values forming a \emph{moment-sequence}; that is, there exists a probability density function $f$ such that $m_i=\widetilde{\mu_i}[f]$ for any $i\in\mathbb{N}$. Then, for any $n\in\mathbb{N}$ and any vector $\vec\alpha=(\alpha_0,\alpha_1,\ldots,\alpha_{n-1})\in(-\infty,2)^{n}$ for which the density
$$
g=\down{\alpha_{n-1}}[\down{\alpha_{n-2}}[...\down{\alpha_{0}}[f]]]
$$
is well defined, we have
$$
\widetilde{M}_{i,\vec\alpha}[g]=m_i, \quad {\rm for \ any} \ i\in\mathbb{N}.
$$
\end{corollary}	
This corollary has the rather significant drawback that, as we have seen in Section \ref{sec:DF}, there seem to be not so many probability density functions allowing for a successive application of the down transformation $n$ times in general. However, there is a much wider class of densities allowing for an application of $n$ successive down transformations on it with suitable exponents $\alpha_i$ for $i=1,2,\ldots,n$. This is why the previous corollary is not so restrictive as it might appear at a first glance.

The following statement, which is derived in the same way as Corollary \ref{cor:MP} from Eq. \eqref{eq:mom_up(n)} and Proposition \ref{prop:inv}, contains an extension of the MaxEnt principle to our newly introduced upper-moments.
\begin{theorem}
		Let $\{m_i\}_{i=1}^N$ be a set of real numbers and let $\mathcal D$ be the set of probability density functions $f$ such that $\sigma_i[f]=m_i$. For some $\alpha<2$ let $\overline \D_\alpha$ be the set of probability densities $g$ such that $m_{i,\alpha}[g]=(m_i)^{\alpha-2}$. Let $\displaystyle f_\lambda=\argmax_{f\in\D}[R_\lambda]$, for $\lambda\in\Rset$. Then, for $\lambda\neq1$ we have
		$$
		\displaystyle \down{\alpha}[f_\lambda]=\argmax_{g\in\overline\D_\alpha}\left[\sigma_{\frac{\lambda-1}{2-\alpha}}^{\frac{1}{\alpha-2}}[f]\right]
		=\argmax_{g\in\overline\D_\alpha}\left[\sigma_{\frac{\lambda-1}{2-\alpha}}^{-1}[f]\right],
		$$
		while for $\lambda=1$ we have
		$$
		\displaystyle \down{\alpha}[f_1]=\argmax_{g\in\overline\D_\alpha}\left[\big\langle\log |x| \big\rangle\right].
		$$
	\end{theorem}
\begin{proof}
	The proof follows directly from Lemma~\ref{lem:MEF}, Eq.~\eqref{eq:mom_up} and the inversion property in Proposition~\ref{prop:inv}.
\end{proof}

\section{Informational inequalities}\label{sec:ineq}

This section contains the most important results of the present work, in form of new informational inequalities for the upper-moments and down-measures we have introduced and explored in the previous section.

\subsection{Upper-moment--moment inequality}\label{sec:upmineq}

The first family of inequalities connects the upper-deviation $\overline{m}_{p,\alpha}$ defined in \eqref{eq:hyperdev} with the classical deviation $\sigma_q$. In order to simplify the following statements, we recall here that, according to Eqs. \eqref{ineq:bip_E-M} and \eqref{ineq:bip_E-M_mirrored}, the entropy-momentum inequality can be written in a compact form as
\begin{equation}\label{ineq:bip_E-M_compact}
\left(\frac{\sigma_{p^*}[f]}{N_{\lambda}[f]}\right)^{\theta_0(p,\lambda)}\geqslant \kappa_{p,\lambda}^{(0)},
\quad \theta_0(p,\lambda)=
\left\{\begin{array}{lcr}
1, &\lambda>\frac1{1+p^*},\qquad  &(\text{classical case})
\smallskip\\
p^*-1,&\lambda<0,\qquad  &(\text{mirrored case})
\end{array}
\right.
\end{equation}
with
$$
\kappa_{p,\lambda}^{(0)}=\begin{cases}
K_{p,\lambda}^{(0)},& \lambda>\frac{1}{1+p^*}
\\
|p^*-1|^{p^*-1}K_{1-\lambda,1-p^*}^{(1)},& \lambda<0
\end{cases}
$$
In order to state our first new informational inequality, we introduce the following family of functions:
\begin{equation}\label{def:gbigstar}
\widetilde g^\bigstar_{p,\alpha,q}=\down{\alpha}[\widetilde g_{p,\lambda}]=\begin{cases}
\down{\alpha}[ g_{p,\lambda}]&\equiv g^\bigstar_{p,\alpha,q},\quad \text{under conditions~\eqref{eq:param_clas}},
\\
\down{\alpha}[g_{1-\lambda,1-p}]&\equiv \overline  g^\bigstar_{p,\alpha,q},\quad \text{under conditions~\eqref{eq:param_mir}}.
\end{cases}
\end{equation}
with
$$
\lambda=\begin{cases}
1+(2-\alpha)q,\quad &\alpha\neq 2,
\\
q,\quad &\alpha=2.
\end{cases}
$$
More explicit forms of the functions defined in Eq. \eqref{def:gbigstar} follow from the general calculations performed in \cite[Proposition 4.1]{IP25} and are given in the Appendix. We can now state and prove an informational inequality relating the upper-moments (or upper-deviations) to classical moments or deviations.
\begin{theorem}[upper-moment--moment inequality]\label{th:hypermom}
Let $p\in\Rset$, $\alpha\in\Rset\setminus\{2\}$ and $q\in\Rset$. Then, for any probability density function, we have

\begin{equation}\label{ineq:hyp-mom}
\left(\frac{m_{p^*,\alpha}[f]}{\sigma_q[f]}\right)^{\Theta_1(p,\alpha,q)}\geqslant \kappa^{(-1)}_{p,\alpha,q}
\end{equation}
with
$$
\Theta_1(p,\alpha,q)=\frac{\theta_0(p,1+(2-\alpha)q)}{\alpha-2},
\quad
\kappa^{(-1)}_{p,\alpha,q}=|2-\alpha|^{\Theta_1(p,\alpha,q)}\kappa^{(0)}_{p,1+(2-\alpha)q},
$$
provided that either in the classical case
\begin{equation}\label{eq:cond1}
p^*=\frac{p}{p-1}>0,\quad \alpha\neq2,\quad (\alpha-2)q<\frac{p^*}{p^*+1}.
\end{equation}
or in the mirrored case
\begin{equation}\label{eq:cond2}
1+q(2-\alpha)<0, \quad \sign[(2-\alpha)q(1+p^*)+p^*]=\sign(p^*-1).
\end{equation}
When $\alpha=2$ and $\lambda\neq1$, the following inequality holds true if one of the conditions \eqref{eq:param_clas} or \eqref{eq:param_mir} is fulfilled:
\begin{equation}\label{ineq:hyp-mom(2)}
\left(M_{p^*,2}[f]^{\frac{1}{p^*}}\sigma^{(E)}_{\lambda-1}[f]\right)^{\theta_0(p,\lambda)}\geqslant \kappa^{(0)}_{p,\lambda},
\end{equation}
Finally, for any $\alpha\in\Rset\setminus\{2\}$ and $p\in\Rset$ such that $p^*>0$, we have

\begin{equation}\label{ineq:hyp-mom(1)}
	\left(\frac{m_{p^*,\alpha}[f]}{\exp\{\mu^{(L)}[f]\}}\right)^\frac1{\alpha-2}\geqslant \kappa_{p,\alpha,0}^{(-1)},
\end{equation}
while for $\alpha=2$ and $p\in\Rset$ such that $p^*>0$, we obtain
\begin{equation}\label{ineq:hyp-mom(12)}
M_{p^*,2}^{\frac{1}{p^*}}[f]\geqslant \kappa^{(0)}_{p,1}e^{<x>}.
\end{equation}
\end{theorem}
\begin{remark}
Alternatively, the inequality~\eqref{ineq:hyp-mom} can be written
	\begin{equation}\label{ineq:hyp-mom_Alter}
	\frac{m_{p^*,\alpha}[f]}{\sigma_q[f]}\left\{\begin{array}{ll}\geqslant\overline{k}, & \xi(p,\alpha,q)>0, \\ \leqslant\overline{k}, & \xi(p,\alpha,q)<0,\end{array}\right.
	\end{equation}
	where
	\begin{equation}\label{eq:hyp-mom_Alter_cond}
\overline k=(\kappa^{(-1)}_{p,\alpha,q})^{\frac1{\Theta(p,\alpha,q)}},\quad \xi(p,\alpha,q)=\frac{p^*}{\alpha-2}+\frac{2-p^*}{\alpha-2}\sign(1+(2-\alpha)q).
	\end{equation}  We also notice that the condition $1+(2-\alpha)q>0$ is equivalent to the classical case of the range of parameters, while the condition $1+(2-\alpha)q<0$ is equivalent to the mirrored range of parameters. Let us stress here that $1+(2-\alpha)q\neq0,$ as it easily follows from Eqs.~\eqref{eq:cond1} and~\eqref{eq:cond2}. It is worth mentioning that in the classical case only the first inequality in~\eqref{ineq:hyp-mom_Alter} occurs.
\end{remark}
\begin{proof}
Assume first that $\alpha\in\Rset\setminus\{2\}$. We start from the moment-entropy inequality Eq. \eqref{ineq:bip_E-M}, which holds true when one of the conditions \eqref{eq:param_clas} or \eqref{eq:param_mir} is satisfied, and we apply it to the up transformed density $\aup$. We find
\begin{equation}\label{eq:interm0}
	\left(\frac{\mu_{p^*}[\aup]^{\frac{1}{p^*}}}{N_{\lambda}[\aup]}\right)^{\theta_0(p,\lambda)}\geqslant \kappa^{(0)}_{p,\lambda}
\end{equation}
or, equivalently, if we assume also that $\lambda\neq1$

\begin{equation}\label{eq:interm1}
\kappa^{(0)}_{p,\lambda}\leqslant\left(\frac{M_{p^*,\alpha}[f]^{\frac{1}{p^*}}}{\left(|2-\alpha|\sigma_{\frac{\lambda-1}{2-\alpha}}[f]\right)^{\frac{1}{\alpha-2}}}\right)^{\theta_0(p,\lambda)}
=\left[\frac{M_{p^*,\alpha}[f]^{\frac{\alpha-2}{p^*}}}{|2-\alpha|\sigma_{\frac{\lambda-1}{2-\alpha}}[f]}\right]^{\frac{\theta_0(p,\lambda)}{\alpha-2}},
\end{equation}
which, taking into account the definition \eqref{eq:hyperdev} and letting $q=(\lambda-1)/(2-\alpha)$ in \eqref{eq:interm1}, gives
$$
\left[\frac{m_{p^*,\alpha}[f]}{\sigma_{q}[f]}\right]^{\frac{\theta_0(p,\lambda)}{\alpha-2}}\geqslant
|2-\alpha|^{\frac{\theta_0(p,1+q(2-\alpha))}{\alpha-2}}\kappa^{(0)}_{p,1+q(2-\alpha)}:=\kappa^{(-1)}_{p,\alpha,q},
$$
leading to the inequality \eqref{ineq:hyp-mom}.

\medskip
\noindent
Concerning the range of parameters, we observe on the one hand that, if the classical conditions \eqref{eq:param_clas} are fulfilled in the moment-entropy inequality, then $p^*>0$ is inherited, while the definition of $q$ and the fact that $\lambda\neq1$ introduces the restriction
$$
\lambda=1+(2-\alpha)q>\frac{1}{1+p^*},
\quad\text{or equivalently,}\quad
(\alpha-2)q<\frac{p^*}{1+p^*}.
$$
Thus, \eqref{ineq:hyp-mom} holds true with the only restrictions $p^*>0$, $\alpha\neq2$ and $(\alpha-2)q<\frac{p^*}{p^*+1}$. On the other hand, if the mirrored conditions \eqref{eq:param_mir} are in force, then we find
$$
1+q(2-\alpha)=\lambda<0,
$$
while the second condition reads
\begin{equation*}
\begin{split}
\sign(1-p^*)&=\sign\left(\frac{q(2-\alpha)}{1+q(2-\alpha)}+p^*\right)=\sign(1+q(2-\alpha))\sign((2-\alpha)q(1+p^*)+p^*)\\
&=-\sign((2-\alpha)q(1+p^*)+p^*),
\end{split}
\end{equation*}
leading to the second condition in Eq. \eqref{eq:cond2}. Regarding minimizing densities, Eq.~\eqref{eq:interm0} is minimized by $f=\widetilde g_{p,\lambda}$. Thus, the minimizing densities $f_{\rm min}$ of Eq.~\eqref{ineq:hyp-mom} must respectively satisfy $(f_{\rm min})^{\uparrow}_\alpha=\widetilde g_{p,\lambda}$ and the proof is complete for $\alpha\neq2$, after applying the down transformation. For $\alpha=2$, it is enough to recall from \cite[Section 3.1]{IP25} that
$$
N_{\lambda}[f^{\uparrow}_2]=\sigma_{\lambda-1}^{(E)}[f]^{-1}
$$
and Eq. \eqref{ineq:hyp-mom(2)} follows immediately from \eqref{eq:interm0}. Let us take next $\lambda=1$ in \eqref{eq:interm0}, recalling that the R\'enyi entropy reduces to the Shannon entropy; more precisely, for $\alpha\in\Rset\setminus\{2\}$, we have
\begin{equation*}
\begin{split}
S[\aup]&=-\int_{\Rset}\aup(u)\log\,\aup(u)\,du=-\int_{\Rset}f(x)\log|(\alpha-2)x|^{\frac{1}{2-\alpha}}\,dx\\
&=\frac{1}{\alpha-2}\int_{\Rset}f(x)\log|(\alpha-2)x|\,dx=\frac{1}{\alpha-2}\int_{\Rset}f(x)\log\,|x|\,dx+\frac{\log|\alpha-2|}{\alpha-2},
\end{split}
\end{equation*}
whence the Shannon entropy power reads
$$
\exp{S[\aup]}=\left(|\alpha-2|\exp\{\mu^{(L)}[f]\}\right)^{\frac{1}{\alpha-2}}.
$$
We then infer from \eqref{eq:interm0}, and the fact that $\theta_0(p,1)=1$, that
$$
\left(\frac{m_{p^*,\alpha}[f]}{\exp\{\mu^{(L)}[f]\}}\right)^\frac1{\alpha-2}\geqslant |\alpha-2|^{\frac1{\alpha-2}}\kappa^{(0)}_{p,1}:=\kappa_{p,\alpha,0}^{(-1)},
$$
which gives \eqref{ineq:hyp-mom(1)}. Since $\lambda=1$, the only domain of the parameters allowed is \eqref{eq:param_clas}, and the condition $p^*>0$ is sufficient for \eqref{ineq:hyp-mom(1)}, since it authomatically implies $\lambda=1>1/(1+p^*)$. Finally, for $\lambda=1$ and $\alpha=2$, we again recall from \cite[Section 3.1]{IP25} that
$$
S[f^{\uparrow}_2]=\int_{\Rset}f(x)x\,dx, \quad {\rm or\, equivalently,} \quad \exp{S[f^{\uparrow}_2]}=e^{<x>},
$$
and \eqref{ineq:hyp-mom(12)} readily follows from \eqref{eq:interm0}. Once more, since we are assuming $\lambda=1$ in the entropy-momentum inequality, we need $p^*>0$ as the unique restriction for \eqref{ineq:hyp-mom(12)} and the proof is complete.
\end{proof}

\begin{remark}[Dealing with generalized Beta probability densities]
Note that the family of probability densities $g_{p,\alpha,q}^{\bigstar}$ is a subfamily of the family of generalized Beta probability densities. This is an interesting fact due to the wide interest of this kind of probability densities in different areas of science.
\end{remark}
\begin{corollary}
When $\Theta_1(p,\alpha,q^*)>0,$ the inequality~\eqref{ineq:hyp-mom} reads
$$
\sigma_{q^*}[f]\leqslant \mathscr{A}\; m_{p^*,\alpha}[f],
$$
with $\mathscr{A}=(\kappa_{p,\alpha,q^*}^{(-1)})^{\frac{1}{\Theta_1(p,\alpha,q^*)}}.$
We then deduce from the classical inequalities that
\begin{equation}
\kappa_{q,\lambda}^{(0)}\leqslant
\frac{\sigma_{q^*}[f]}{N_\lambda[f]}
\leqslant
\mathscr{A}\,
 \frac{m_{p^*,\alpha}[f]}{N_\lambda[f]},
\end{equation}
and
\begin{equation}
\widehat\kappa_{q,\lambda}^{(0)}\leqslant
\phi_{q,\lambda}[f]\sigma_{q^*}[f]
\leqslant
\mathscr{A}\,m_{p^*,\alpha}[f]\phi_{q,\lambda}[f].
\end{equation}
\end{corollary}

The inequality \eqref{ineq:hyp-mom} involving the first order upper-moment for $\alpha\neq2$ can be generalized to higher-order upper-moments. More precisely, we have
\begin{theorem}\label{th:hypergen}
Let $\vec\alpha=(\alpha_0,\alpha_1,\ldots,\alpha_{n})$ be a vector such that $\alpha_i\in\Rset\setminus\{2\}$, for $i=0,1,\ldots,n$. Denote by $\vec\alpha'=(\alpha_1,\ldots,\alpha_{n})$. Then, for any probability density $f$ and for any $p$, $q\in\Rset$ such that $(p,\alpha_0,q)$ satisfy the conditions given in Eqs. \eqref{eq:cond1} or \eqref{eq:cond2}, there exists $k\in\Rset$ depending on $(p,q,\vec\alpha)$ such that
\begin{equation}\label{ineq:hyp-mom-gen}
\frac{m_{p^*,\vec\alpha}[f]}{m_{q,\vec\alpha'}[f]}\geqslant k, \quad {\rm respectively}, \quad \frac{m_{p^*,\vec\alpha}[f]}{m_{q,\vec\alpha'}[f]}\leqslant k,
\end{equation}
the first inequality in Eq. \eqref{ineq:hyp-mom-gen} being in force if $\vec\alpha'$ contains an even number of elements $\alpha_i<2$, while the second inequality in Eq. \eqref{ineq:hyp-mom-gen} being in force if $\vec\alpha'$ contains an odd number of elements $\alpha_i<2$, provided that $\xi(p,\alpha_0,q)>0,$ and vice-versa if $\xi(p,\alpha_0,q)<0$, where the function $\xi$ is defined in Eq.~\eqref{eq:hyp-mom_Alter_cond}.
\end{theorem}
\begin{proof}
The proof follows easily by induction on $n$. In order to avoid complicating the notation we give here only the step $n=2$, which is sufficient to give the pattern leading to the inequality Eq. \eqref{ineq:hyp-mom-gen} for general vectors of any length. Let thus $p$, $q$ and $\vec\alpha=(\alpha_0,\alpha_1)$ as in the statement of Theorem \ref{th:hypergen}. Assume first that $\xi(p,\alpha_0,q)>0$, thus Theorem \ref{th:hypermom} gives
$$
\frac{m_{p^*,\alpha_0}[f]}{\sigma_q[f]}\geqslant\overline{k}.
$$
We apply the previous inequality to $f=f^{\uparrow}_{\alpha_1}$ to obtain that
$$
\overline{k}\leqslant\frac{M_{p^*,\alpha_0}[f^{\uparrow}_{\alpha_1}]^{\frac{\alpha_0-2}{p^*}}}{\sigma_q[f^{\uparrow}_{\alpha_1}]}
=\frac{M_{p^*,(\alpha_0,\alpha_1)}[f]^{\frac{\alpha_0-2}{p^*}}}{M_{q,\alpha_1}[f]^{1/q}}.
$$
We raise the previous inequality to the power $\alpha_1-2$ to get
$$
\frac{m_{p^*,\vec\alpha}[f]}{m_{q,\alpha_1}[f]}\geqslant k:=\overline{k}^{\alpha_1-2},
$$
provided $\alpha_1>2$, or the opposite inequality if $\alpha_1<2$. Observe also that the inequality signs are reversed from the starting point if  $\xi(p,\alpha_0,q)<0$, which completes the proof of the case $n=2$. The general case follows in an obvious way by induction.
\end{proof}

\medskip

\noindent \textbf{Minimizers.} As it is well-known, the entropy-momentum inequality Eq. \eqref{ineq:bip_E-M} has the stretched Gaussian functions $g_{p,\lambda}$ as minimizers. Since the inequalities in Theorems \ref{th:hypermom} and \ref{th:hypergen} are obtained by applying Eq. \eqref{ineq:bip_E-M} to $\aup$ for suitable values of $\alpha$, we deduce from Proposition \ref{prop:inv} that their minimizers are obtained by applying successively the down transformations to the functions $g_{p,\lambda}$. Thus, for the inequality Eq. \eqref{ineq:hyp-mom} and similar ones at first order, the minimizer is $\mathcal{D}_{\alpha}[g_{p,\lambda}]$ and has been carefully calculated in our previous paper \cite[Proposition 4.1]{IP25}. For the general inequality Eq. \eqref{ineq:hyp-mom-gen}, the minimizer is given by
$$
f_{{\rm min}}^{(n)}(\vec\alpha):=\mathcal{D}_{\alpha_n}[\mathcal{D}_{\alpha_{n-1}}[\ldots[\mathcal{D}_{\alpha_0}[g_{p,\lambda}]]]],
$$
whenever the previous expression is well defined. As discussed in~\cite[Section 2]{IP25}, in order to apply $\down{\alpha_1}$ to $\down{\alpha_0}[g_{p,\lambda}]$, it is necessary that $\alpha_0>\sup\limits_{x\in\mathbb R}\left(\frac{g_{p,\lambda}\,g_{p,\lambda}''}{(g_{p,\lambda}')^2}\right)$, which does not hold true for all $\alpha_0\in \Rset.$ In fact, straightforward computations lead to
$$
\frac{g_{p,\lambda}\,g_{p,\lambda}''}{(g_{p,\lambda}')^2}=2-\frac{\lambda}{p^*}-\frac1p-\frac1{p^*x^{p^*}}.
$$
As a consequence, we can ensure that $\down{\alpha_1}[\down{\alpha_0}[g_{p,\lambda}]]$ is well defined only if $\alpha_0>2-\frac{\lambda}{p^*}-\frac1p$; therefore, the corresponding inequality is sharp. We conjecture that in all the other cases the inequalities remain sharp, but a rigorous proof (if any) will be addressed in a future work.

\medskip

\noindent \textbf{Remark.} A simple counterexample of non applicability of $\down{\beta}\down{\alpha}$ can be find applying $\down{1}$ to the exponential density, leading with a uniform density. As the derivative of the uniform density is identically equal to zero, it is not possible to apply a second down transformation $\down{\alpha}$ for any $\alpha\in\Rset.$

\medskip

Of course, calculating such a minimizer is not an easy task, since already for a single step the expressions in \cite[Proposition 4.1]{IP25} were far from being trivial. However, as noticed in \cite[Corollary 4.1]{IP25}, there is a particular case when we recover stretched Gaussian functions as minimizers; that is, if $\alpha_0=2-\lambda$, then
$$
f_{{\rm min}}^{(1)}(2-\lambda)=\mathcal{D}_{2-\lambda}[g_{p,\lambda}]=g_{1-\lambda,1-p}.
$$
If then $\vec\alpha=(\alpha_0,\alpha_1)$ with $\alpha_0=2-\lambda$, $\alpha_1=1+p$, we get
$$
f_{{\rm min}}^{(2)}(\vec\alpha)=\mathcal{D}_{1+p}[g_{1-\lambda,1-p}]=g_{p,\lambda},
$$
and this particular case can be generalized to vectors of any length in an obvious way, alternating their components between $2-\lambda$ and $1+p$ in order to obtain as minimizers either $g_{1-\lambda,1-p}$ or $g_{p,\lambda}$. Additionally, the latter fact allows us to find the minimizing density for any vector $\vec \alpha^{\bigstar}$ with one of the following forms
$$
\vec \alpha_1^\bigstar=(2-\lambda,1+p,2-\lambda,1+p,\ldots,1+p, \alpha_n),\quad \vec \alpha_2^\bigstar=(2-\lambda,1+p,2-\lambda,1+p,\ldots,2-\lambda, \alpha_m).
$$
More precisely,
$$
f^{\rm min}_{\vec\alpha_1^\bigstar}=\down{\alpha_n}[g_{p,\lambda}]=g^\bigstar_{p,\alpha_n,\frac{1-\lambda}{2-\alpha_n}},\quad f^{\rm min}_{\vec\alpha_2^\bigstar}=\down{\alpha_m}[g_{1-\lambda,1-p}]=g^\bigstar_{1-\lambda,\alpha_m,\frac{p}{2-\alpha_m}}.
$$

\subsection{Upper-moment--entropy inequality}\label{sec:upment}

In this rather short subsection, we give another new informational inequality, which relates the newly introduced upper-moments or, more precisely, the upper-deviation defined in \eqref{eq:hyperdev}, with the entropy power.

\begin{theorem}[Upper-moment--entropy inequality]\label{th:EM}
Let $p\geqslant1$ and $2-\alpha_0=\lambda>1/(1+p^*)$ (or equivalently $\alpha_0<\frac{2p^*+1}{p^*+1}=\frac{3p-1}{2p-1}$) be real numbers. Then, for any probability density $f$ such that $\up{\alpha_0}[f]$ is absolutely continuous if $p>1$ or of bounded variation if $p=1$, we have
\begin{equation}\label{ineq:ent-hyp-mom}
\frac{N_{1-p}[f]}{m_{p^{*},\alpha_0}[f]}\geqslant \widehat\kappa_{p,\alpha_0}^{(-1)}:=\left(\widehat\kappa_{p,2-\alpha_0}^{(0)}\right)^{2-\alpha_0}
\end{equation}
The minimizing densities for Eq. \eqref{ineq:ent-hyp-mom} are given by $f_{\rm min}=g_{\alpha_0-1,1-p}.$ Moreover, if $\alpha_1>2$, we also have the generalized inequality
\begin{equation}\label{ineq:ent-hyp-mom-gen}
\frac{\sigma_{\frac{p}{\alpha_1-2}}[f]}{m_{p^*,(\alpha_0,\alpha_1)}[f]}\geqslant\widehat\kappa_{p,\alpha_0,\alpha_1}^{(-2)}
:=\frac{\left(\widehat\kappa_{p,2-\alpha_0}^{(0)}\right)^{(2-\alpha_0)(\alpha_1-2)}}{|\alpha_1-2|},
\end{equation}
which is in force for any probability density $f$ such that $\up{\alpha_0}[\up{\alpha_1}[f]]$ is absolutely continuous if $p>1$ or of bounded variation if $p=1$. If $\alpha_1<2$, the inequality sign in Eq. \eqref{ineq:ent-hyp-mom-gen} is reversed. The minimizing densities for Eq. \eqref{ineq:ent-hyp-mom-gen} are given by $f_{\rm min}=\overline g^{\bigstar}_{p,\alpha_1,\frac{1-\alpha_0}{2-\alpha_1}}.$
\end{theorem}
\begin{proof}
We start from the generalized Cram\'er-Rao inequality Eq. \eqref{ineq:CR}, which holds true in the conditions $p\geqslant1$ and $\lambda>1/(1+p^*)$, and apply it to the up transformed density $\up{\alpha_0}[f]$ according to~\cite[Theorem 5]{Lutwak05} to deduce that
$$
\sigma_{p^*}[\up{\alpha_0}[f]]\phi_{p,\lambda}[\up{\alpha_0}[f]]\geqslant \widehat\kappa_{p,\lambda}^{(0)}:=\kappa_{p,\lambda}^{(0)}\kappa_{p,\lambda}^{(1)}.
$$
We particularize $\alpha_0=2-\lambda$ in the previous inequality and find, with the aid of Lemma \ref{lem:MEF}, that
$$
M_{p^*,2-\lambda}[f]^{\frac{1}{p^*}}N_{1-p}[f]^{\frac{1}{\lambda}}\geqslant\widehat\kappa_{p,\lambda}^{(0)}.
$$
The positivity of $\lambda$ then entails
$$
M_{p^*,2-\lambda}[f]^{\frac{\lambda}{p^*}}N_{1-p}[f]\geqslant\widehat \kappa_{p,2-\lambda}^{(-1)}:=\left(\widehat\kappa_{p,\lambda}^{(0)}\right)^{\lambda},
$$
and the inequality Eq. \eqref{ineq:ent-hyp-mom} follows by noticing that $\lambda/p^{*}=(2-\alpha_0)/p^{*}$ and recalling the definition of the upper-deviation \eqref{eq:hyperdev}. The minimizing densities satisfy $f_{\rm min}=\down{2-\lambda}[g_{p,\lambda}]=g_{1-\lambda,1-p}=g_{\alpha_0-1,1-p}$.

For the generalized inequality Eq. \eqref{ineq:ent-hyp-mom-gen}, we pick $\alpha_1>2$ as in the statement and we apply Eq. \eqref{ineq:ent-hyp-mom} to $\up{\alpha_1}[f]$, deducing that
$$
N_{1-p}[\up{\alpha_1}[f]]M_{p^{*},2-\lambda}^{\frac{\lambda}{p^{*}}}[\up{\alpha_1}[f]]\geqslant \widehat \kappa_{p,2-\lambda}^{(-1)}.
$$
We then infer from Lemma \ref{lem:MEF} and Definition \ref{def:hypergen} that
$$
N_{1-p}[\up{\alpha_1}[f]]=\left(|2-\alpha_1|\sigma_{\frac{p}{\alpha_1-2}}[f]\right)^{\frac{1}{\alpha_1-2}}, \quad
M_{p^{*},2-\lambda}^{\frac{\lambda}{p^{*}}}[\aup]=M_{p^{*},(2-\lambda,\alpha_1)}[f]^{\frac{\lambda}{p^{*}}},
$$
hence, taking into account the positivity of $\alpha_1-2$, we deduce that
\begin{equation}\label{eq:interm4}
\sigma_{\frac{p}{\alpha_1-2}}[f]M_{p^{*},(2-\lambda,\alpha_1)}[f]^{\frac{\lambda(\alpha_1-2)}{p^{*}}}
\geqslant \widehat \kappa_{p,2-\lambda,\alpha_1}^{(-2)}:=\frac{(\widehat\kappa_{p,2-\lambda}^{(-1)})^{\alpha_1-2}}{|\alpha_1-2|}
\end{equation}
The inequality Eq. \eqref{ineq:ent-hyp-mom-gen} follows then from \eqref{eq:interm4} by noticing that
$$
m_{p^{*},(2-\lambda,\alpha_1)}[f]=M_{p^{*},(2-\lambda,\alpha_1)}[f]^{-\frac{\lambda(\alpha_1-2)}{p^{*}}},
$$
according to Eq. \eqref{eq:hyperdevgen}. The fact that the inequality sign is reversed if we pick $\alpha_1<2$ is obvious from the previous steps. Finally, the minimizing densities must satisfy $$f_{\rm min}=\down{\alpha_1}[g_{1-\lambda,1-p}]=\overline g^{\bigstar}_{p,\alpha_1,\frac{\lambda-1}{2-\alpha_1}}.$$
\end{proof}

\begin{remark}
The minimizers for the inequalities \eqref{ineq:ent-hyp-mom} and \eqref{ineq:ent-hyp-mom-gen} are once again obtained as down transformations applied to the stretched Gaussians $g_{p,\lambda}$, calculated in \cite[Proposition 4.1]{IP25}, and the discussion given at the end of the previous section applies here too. We also notice that the previous process can be further iterated to vectors of higher length than $(\alpha_0,\alpha_1)$.
\end{remark}

\medskip

\noindent \textbf{Open problem.} In the statement of Theorem \ref{th:EM} we have imposed some regularity conditions to the density $\up{\alpha_0}[f]$ in order for the Cramér-Rao inequality to hold true, conditions stemming from \cite[Theorem 5]{Lutwak05}. We believe that it is an interesting open question to derive sharp (as much as possible) classes of probability densities $f$ such that $\up{\alpha_0}[f]$ is absolutely continuous (or of bounded variation), as needed. We refrain from discussing this interesting question in the current paper, noticing that for the minimizers and for "easy" densities (such as powers, decreasing exponentials) the condition is fulfilled in a range of values of $\alpha$.`

\begin{corollary}
From inequality~\eqref{ineq:ent-hyp-mom} follows
\begin{equation*}
\frac{1}{N_{1-p}[f]}\leqslant \frac{\mathscr B}{m_{p^*,\alpha}[f]},\quad \mathscr B=\left(\kappa^{(-1)}_{p,\alpha}\right)^{-1}.
	\end{equation*}
Consequently
\begin{equation*}
\kappa^{(0)}_{q,1-p}
\leqslant
\frac{\sigma_{q^*}[f]}{N_{1-p}[f]}\leqslant \frac{\mathscr B\,\sigma_{q^*}[f]}{m_{p^*,\alpha}[f]}.
\end{equation*}
\end{corollary}

\subsection{Down-Fisher--Fisher inequality}\label{sec:dffineq}

In this section we introduce one more informational inequality involving the new functionals defined in the previous sections. This inequality connects the down-Fisher measure introduced in Definition \ref{def:sub-fisher} to the classical Fisher information \eqref{eq:def_FI}. Before stating and proving the main result, we introduce the following function, which will be useful when describing the minimizers of the inequality:
\begin{equation}\label{eq:gdagger}
\widetilde g_{p,q,s,\lambda}^{\dagger}=\up{2-\lambda}[\widetilde g_{p,\frac{p-q}{p},1-s}].
\end{equation}
We have
\begin{theorem}[Down-Fisher--Fisher inequality]
Let $p$, $q$, $s$ and $\lambda$ be real numbers such that $p\neq q$ and $\lambda\neq0$. Assume that $(p,q,s,\lambda)$ satisfy one of the following conditions:

$\bullet$ either $p\geqslant1$ and
\begin{equation}\label{eq:sffcond1}
\sign\left(\frac{p-q}{p-1}-s\right)=\sign\left(\frac{p-q}{p}+s\right)\neq0,
\end{equation}

$\bullet$ or $p<1$ and
\begin{equation}\label{eq:sffcond2}
\sign\left(\frac{p-q}{p-1}-s\right)=\sign\left(\frac{p-q}{p}-s\right)\neq0,
\quad
\sign\left(s\right)=\sign\left(\frac{p-q}{p-1}\right)\neq0.
\end{equation}

Then, for any probability density function $f$ such that $\down{2-\lambda}[f]$ is absolutely continuous, we have
\begin{equation}\label{ineq:sff}
\left[F_{s,\lambda}[f]^{\frac{1}{s}}\varphi_{p,q,r}[f]^{\frac{1}{p-q}}\right]^{\theta_2(s,p,q)}
	\geqslant \kappa^{(2)}_{s,p,q}:=\kappa^{(1)}_{p,\frac{p-q}{p},1-s}, \quad r:=\frac{(2-\lambda)(p-q)}{p},
\end{equation}
where
\begin{equation}\label{eq:theta2}
\theta_2(s,p,q)=\left\{\begin{array}{ll}1-\frac{q}{p}+s, & {\rm if} \ p\geqslant1,\\[1mm]\frac{q}{p}-1, & {\rm if} \ p<1.\end{array}\right.
\end{equation}
The minimizing densities are given by $f_{\rm min}=\widetilde g_{p,q,s,\lambda}^{\dagger}$.
\end{theorem}
\begin{proof}
We start from the extended (tri-parametric) Stam inequality given by Eq. \eqref{ineq:trip_Stam_extended} in the classical domain and in the mirrored domain, recalling that, under some suitable conditions (such as \eqref{eq:sign_cond1} if $p\geqslant1$ and \eqref{eq:sign_cond2} if $p<1$) on the parameters $(p,\beta,\lambda)$, we have
\begin{equation}\label{ineq:tripStam}
\left[N_{\lambda}[f]\phi_{p,\beta}[f]\right]^{\theta_1(p,\beta,\lambda)}\geqslant \kappa^{(1)}_{p,\beta,\lambda}, \quad \phi_{p,\beta}[f]=F_{p,\beta}[f]^{\frac{1}{p\beta}},
\end{equation}
where
\begin{equation}\label{eq:theta}
\theta_1(p,\beta,\lambda)=\left\{\begin{array}{ll}1+\beta-\lambda, & {\rm if} \ p\geqslant1,\\-\beta, & {\rm if} \ p<1.\end{array}\right.
\end{equation}
For $\alpha\in\Rset$ such that $\adown$ is absolutely continuous, we apply the previous inequality to the density $\adown$. Recalling that, on the one hand, Lemma \ref{lem:MEF} gives
$$
N_{\lambda}[\adown]=\phi_{1-\lambda,2-\alpha}[f]^{2-\alpha}=F_{1-\lambda,2-\alpha}[f]^{\frac{1}{1-\lambda}},
$$
while, on the other hand, Lemma \ref{lem:sub-fisher} ensures that
$$
\phi_{p,\beta}[\adown]=\varphi_{p,p(1-\beta),\alpha\beta}[f]^{\frac{1}{p\beta}},
$$
we obtain from Eq. \eqref{ineq:tripStam} that
\begin{equation}\label{eq:interm5}
\left[F_{1-\lambda,2-\alpha}[f]^{\frac{1}{1-\lambda}}\varphi_{p,p(1-\beta),\alpha\beta}[f]^{\frac{1}{p\beta}}\right]^{\theta(\beta,\lambda)}\geqslant \kappa^{(1)}_{p,\beta,\lambda}.
\end{equation}
We next introduce the following notation
\begin{equation}\label{eq:interm6}
s:=1-\lambda, \quad \widetilde{\lambda}:=2-\alpha, \quad q=p(1-\beta)
\end{equation}
and observe, similarly as in the proof of Lemma \ref{lem:sub-fisher}, that
$$
\alpha\beta=\frac{(2-\widetilde{\lambda})(p-q)}{p}, \quad p\beta=p-q.
$$
Moreover, an easy calculation starting from Eq. \eqref{eq:theta} shows that $\theta_1(p,\beta,\lambda)$ writes in the new notation Eq. \eqref{eq:interm6} exactly as the expression of $\theta_2(s,p,q)$ given in Eq. \eqref{eq:theta2}. Thus, Eq. \eqref{eq:interm5} writes
$$
\left[F_{s,\widetilde{\lambda}}[f]^{\frac{1}{s}}\varphi_{p,q,r}[f]^{\frac{1}{p-q}}\right]^{\theta_2(s,p,q)}\geqslant \kappa^{(1)}_{p,\beta,\lambda}, \quad r=\frac{(2-\widetilde{\lambda})(p-q)}{p}.
$$
The inequality Eq. \eqref{ineq:sff} follows from the previous one by renaming the parameter $\lambda$ instead of $\widetilde{\lambda}$. Minimizing densities are obtained from $\down{\alpha}[f_{\rm min}]=\widetilde g_{p,\beta,\lambda},$ which leads to $f_{\rm min}=\widetilde g_{p,q,s,\lambda}^{\dagger}$, according to \eqref{eq:gdagger}. With respect to the range of parameters, we recall that the triparametric Stam inequality \cite[Theorem 5.1]{IP25} holds true for either $p\geqslant1$ and
\begin{equation}\label{eq:cond1trip}
\sign(p^*\beta+\lambda-1)=\sign(\beta+1-\lambda)\neq0,
\end{equation}
or for $p<1$ and the conditions
\begin{equation}\label{eq:cond2trip}
\sign(p^*\beta+\lambda-1)=\sign(\beta-1+\lambda)\neq0, \quad \sign(\lambda-1)=\sign(\beta)\neq0.
\end{equation}
It is straightforward to check that, in the notation introduced in Eq. \eqref{eq:interm6}, Eq. \eqref{eq:cond1trip} transforms into
$$
\sign\left(p^*\left(1-\frac qp\right)-s\right)=\sign\left(1-\frac qp+s\right),
$$
leading to \eqref{eq:sffcond1}. Analogously, Eq. \eqref{eq:cond2trip} transforms into
$$
\sign\left(p^*\left(1-\frac qp\right)-s\right)=\sign\left(1-\frac qp-s\right),
\quad
\sign(-s)=\sign\left(1-\frac qp\right),
$$
which readily gives \eqref{eq:sffcond2}.
\end{proof}

\begin{corollary}
From inequality~\eqref{ineq:sff}, when $\theta_2(s,p,q)<0,$ follows
\begin{equation}
\phi_{s,\lambda}[f]\leqslant
\mathscr C\; \varphi_{p,q,r}^{\frac{1}{\lambda(q-p)}}[f],\quad \mathscr C=\left(\kappa_{s,p,q}^{(2)}\right)^{\frac{1}{\lambda\theta_2(s,p,q)}}.
\end{equation}
Then
\begin{equation}
\widehat\kappa_{s,\lambda}^{(0)}
\leqslant
\sigma_{s^*}[f]\phi_{s,\lambda}[f]\leqslant
\mathscr C\;
\sigma_{s^*}[f]\varphi_{p,q,r}^{\frac{1}{q-p}}[f]
\end{equation}
and
\begin{equation}
\kappa_{s,\lambda,\overline{\lambda}}^{(1)}
\leqslant
N_{\overline\lambda}[f]\phi_{s,\lambda}[f]\leqslant
\mathscr C\;
N_{\overline\lambda}[f]\varphi_{p,q,r}^{\frac{1}{q-p}}[f]
\end{equation}
\end{corollary}
\begin{remark}
From the definition~\eqref{eq:gdagger} of the functions $\widetilde g_{p,q,\frac qp,1-p}^{\dagger}$ and the definition of the functions $\widetilde g_{p,\beta,\lambda}$ in~\cite[Section 5]{IP25}, it follows after some algebraic manipulation that
$$
\widetilde g_{p,q,\frac qp,\lambda}^{\dagger}=\widetilde g_{\frac qp,\,\lambda,\,1+p^*\lambda}=\up{2-\lambda}\left[g_{p,\frac{p-q}{p}}\right].
$$
That is to say, the family of probability densities $g^\dagger_{}$ encompasses the minimizing densities of the mirrored domain of the triparametric Stam inequality~\cite[Section 5]{IP25}. Moreover, in the case $\lambda=1-p$ one finds
$$
\widetilde g_{p,q,\frac qp,1-p}^{\dagger}= g_{\frac qp,\,1-p}.
$$
\end{remark}
\begin{corollary}[Down-Fisher--Entropy inequality]
	In the case $\beta=\lambda$, Eq.~\eqref{eq:interm6} implies $s=\frac qp$, then for $p\geqslant1$ the inequality~\eqref{ineq:sff} reads
	$$
	F_{s,\lambda}[f]^{\frac{1}{s}}\varphi_{p,q,r}[f]^{\frac{1}{p-q}}\geqslant \kappa^{(2)}_{s,p,q}.
	$$
	Moreover, the generalized Stam inequality~\eqref{ineq:trip_Stam_extended} with parameters $(s,\lambda,\lambda)$ in the mirrored domain $s<1$  reads
	$F_{s,\lambda}[f]^{-\frac 1{s}} N_\lambda[f]^{-\lambda}\geqslant \kappa^{(1)}_{s,\lambda}.$
	The product of the latter inequalities gives
	\begin{equation}
	\frac{\varphi_{p,q,r}[f]^{\frac{1}{p-q}}}{N_\lambda[f]^{\lambda}}\geqslant \widehat \kappa^{(1)}_{s,p,q,\lambda}:=\kappa^{(1)}_{s,\lambda}\kappa^{(2)}_{s,p,q}
		\end{equation}
Once again, the latter inequality allows us to write
	\begin{equation}
N_{\lambda}[f]
\leqslant
\mathscr D\varphi_{p,q,r}[f]^{\frac{1}{\lambda(p-q)}},\quad \mathscr D=\left(\widehat{\kappa}_{s,p,q,\lambda}^{(1)}\right)^{-1/\lambda}.
    \end{equation}	
Moreover,
   \begin{equation}
\kappa_{p,\beta,\lambda}^{(1)}\leqslant
N_{\lambda}[f]\phi_{p,\beta}[f]
\leqslant
\mathscr D\varphi_{p,q,r}[f]^{\frac{1}{\lambda(p-q)}}\phi_{p,\beta}[f].
   \end{equation}
\end{corollary}
To conclude this section, we give an inequality involving down-Fisher measures, Shannon entropy and the expected value of the logarithm of the derivative of the probability density.
\begin{theorem} Let $f$ be a decreasing probability density function and $\alpha\in\Rset$ such that $\adown$ is absolutely continuous. Let $p,q,\beta$ real numbers such that $p\geqslant1$ and $q\neq p.$ Then
	\begin{equation}\label{eqth:subFish-Shannon}
	\left(\varphi_{p,q,\beta}[f]\right)^{\frac{1}{p}}e^{\beta S[f]}e^{\frac{p-q}p\left\langle\log|f'|\right\rangle}\geqslant\kappa^{(1)}_{p,\frac{p-q}p,1}.
	\end{equation}
Moreover, the equality holds only when $f=\up{\frac{\beta p}{p-q}}[g_{p,\frac{p-q}p,1}]$
\end{theorem}
\begin{proof}
Let $f$ be a decreasing probability density function and $\alpha\in\Rset$ such that $\adown$ is absolutely continuous. We deduce from the inequality~\eqref{ineq:tripStam} with $\lambda=1$, $\beta\neq0$ and $p\geqslant1$ that
$$
(\phi_{p,\beta}[\adown]e^{S[\adown]})^{\beta}\geqslant \kappa^{(1)}_{p,\beta,1},
$$
hence it follows from Lemmas~\ref{lem:MEF} and~\ref{lem:sub-fisher} that
\begin{equation}
\varphi_{p,q,\overline\beta}[f]^{\frac{1}{p}}e^{\overline{\beta}S[f]}e^{\beta\left\langle\log|f'|\right\rangle}\geqslant \kappa^{(1)}_{p,\beta,1},
\end{equation}
with $q=p(1-\beta)$ and $\overline \beta=\alpha\beta$, leading to inequality~\eqref{eqth:subFish-Shannon} after replacing $\beta=\frac{p-q}p$ and simple algebraic manipulations. Finally, we note that $\beta\neq0$ implies $q\neq p,$ concluding the proof.
\end{proof}

\subsection{Upper bounds for the generalized Stam product}\label{sec:dfshineq}

In this section, we derive further inequalities involving the Shannon entropy power, as well as other informational functionals depending on the first and second derivatives of the probability density function. In particular, the latter ones will allow us to obtain, as a simple consequence, upper bounds for the generalized (and also classical) Stam product. We start with some bounds involving the Shannon entropy, the generalized Fisher information and the first derivative of the density function.
\begin{proposition}\label{prop:bounds}
Let $p>0$ and $\lambda<2$ be real numbers, and let $f$ be a decreasing probability density. Then
	\begin{equation}\label{eq:bounds_Stam}
	\kappa_{p,\lambda,1}^{(1)}
	\leqslant
	\phi_{p,\lambda}[f]e^{S[f]}
	\leqslant
	e^{-\frac1{2-\lambda}\langle\log|f'|\rangle}\phi_{p,\lambda}^{\frac{2}{2-\lambda}}[f]
	\end{equation}
provided that the involved quantities are finite.
\end{proposition}
\begin{proof}
We apply the classical inequality
	$$
	N_\lambda[f]e^{-S[f]}\geqslant1,\quad \lambda<1
	$$
to the down transformed density $\adown$ and we infer from Lemma \ref{lem:MEF} that 	
	$$
	N_\lambda[\adown]e^{-S[\adown]}=\phi_{1-\lambda,2-\alpha}^{2-\alpha}[f]e^{-\alpha S[f]}e^{-\langle\log|f'|\rangle}\geqslant 1.
	$$
Assuming that $\alpha>0$, we further obtain
	$$
	e^{S[f]}
	\leqslant
	e^{-\frac1{\alpha}\langle\log|f'|\rangle}\phi_{1-\lambda,2-\alpha}^{\frac{2-\alpha}{\alpha}}[f]
	$$
	and straightforwardly
	$$
	\kappa_{1-\lambda,2-\alpha,1}^{(1)}
	\leqslant
	\phi_{1-\lambda,2-\alpha}[f]e^{S[f]}
	\leqslant
	e^{-\frac1{\alpha}\langle\log|f'|\rangle}\phi_{1-\lambda,2-\alpha}^{\frac{2}{\alpha}}[f],
	$$
which is equivalent to Eq. \eqref{eq:bounds_Stam} after adopting the notation $p=1-\lambda$, $\widetilde\lambda=2-\alpha$.	
\end{proof}
\begin{corollary} In the particular case $p=2$ and $\lambda=1$ we obtain an upper bound for the classical Stam product
	\begin{equation}
	\kappa_{2,1}^{(1)}
	\leqslant
	\sqrt{F[f]}e^{S[f]}
	\leqslant
	e^{-\langle\log|f'|\rangle}F[f]
	\end{equation}
\end{corollary}
The last theorem of this paper introduces a new Stam-like inequality involving the Shannon entropy power, the Fisher information and the second derivative of the density function, and is obtained by applying twice the down transformation to the moment-entropy inequality.
\begin{theorem}[modified Stam-like inequality]
Let $\lambda$ be a real number and let $f$ be a probability density such that
$$
\sup\limits_{x\in\Rset}\left(\frac{f(x)f''(x)}{(f')^2(x)}\right)<2-\lambda.
$$
Then, for any $\beta\in\Rset$ and $p$ such that $\sign(p)=\sign(\beta-2)\neq 0,$ we have the following inequality
 \begin{equation}\label{eqth:mod_Stam}
\left(\phi_{p,\lambda}[f]N[f]^{\omega(\beta,\lambda)}\right)^{(\beta-2)\lambda}
\geqslant
 \widetilde \kappa_{p,\beta} e^{(\beta-1)\left\langle \log|f'|\right \rangle} \exp\left\langle\log\left(2-\lambda-\frac{ff''}{(f')^2}\right)\right\rangle,
 \end{equation}
where
$$
N[f]=e^{S[f]}, \quad \omega(\beta,\lambda)=1+\frac{2-2\beta}{\lambda(\beta-2)}, \quad \widetilde \kappa_{p,\beta}=|2-\beta|\kappa^{(0)}_{p,1}.
$$
When $\beta=2$ we have
\begin{equation}\label{eqth:mod_Stam2}
\left\langle\left|\log\frac{f^{2-\lambda}}{|f'|}\right|^{p^*}\right\rangle^{\frac{1}{p^*}}e^{-2S[f]}
\geqslant \kappa_{p,1}^{(0)} e^{\left\langle \log|f'|\right \rangle} \exp\left\langle\log\left(2-\lambda-\frac{ff''}{(f')^2}\right)\right\rangle.
\end{equation}
Moreover, the minimizing densities for \eqref{eqth:mod_Stam} and \eqref{eqth:mod_Stam2} are given by
$$
f_{\rm min}=\up{2-\lambda}\up\beta\left[g_{\frac{p}{2+p-\beta},1}\right]\quad \beta\neq2,\quad {\rm respectively}
	\quad f_{\rm min}=\up{2-\lambda}\up2\left[g_{p,1}\right],\quad \beta=2.
$$
Finally, given $(p,\lambda)$ satisfying conditions~\eqref{eq:param_clas} or~\eqref{eq:param_mir}, the following inequality holds true:	
\begin{equation}\label{eqth:subF-further}
 \left(
 \left\langle\left|\log\frac{f(x)^{2-\lambda}}{|f'(x)|}\right|
 ^{p^*}\right\rangle^{\frac{1}{p^*}}
 \varphi_{1-\lambda,1-\lambda,0}^\frac{1}{\lambda-1}[f]
 \right)^{\theta_0(p,\lambda)}
 \geqslant \kappa^{(0)}_{p,\lambda}.
\end{equation}
The minimizing density for \eqref{eqth:subF-further} is given by $f_{\rm min}=\up{2-\lambda}\up2[\widetilde g_{p,\lambda}]$.
\end{theorem}
\begin{proof}
We first let $\beta\neq2$. Given $p^*>0$, we derive from the moment-entropy inequality that
	$$
\frac{\sigma_{p^*}[f]}{e^{S[f]}}\ge \kappa^{(0)}_{p,1}.
	$$
Assuming that $\alpha\in\Rset$ is such that $\sup\left(\frac{f f''}{(f')^2}\right)<\alpha$, we can apply the latter inequality to $f^{\downarrow\downarrow}_{\alpha\beta}$. According to Lemma \ref{lem:MEF}, we get
\begin{equation}\label{eq:downdown_sigma}
\sigma_{p^*}[f^{\downarrow\downarrow}_{\alpha\beta}]
=
\frac{N_{1+(2-\beta)p^*}^{\beta-2}[\adown]}{|2-\beta|}
=
\frac{\phi_{(\beta-2)p^*,2-\alpha}^{(\beta-2)(2-\alpha)}[f]}{|2-\beta|},
\end{equation}
and we deduce from Eqs.~\eqref{eq:downdown_sigma} and \eqref{eqprop:downdown} that
\begin{equation}
\phi_{(\beta-2)p^*,2-\alpha}[f]^{(\beta-2)(2-\alpha)}e^{(2\alpha-\alpha\beta-2)S[f]}\geqslant |2-\beta|\kappa^{(0)}_{p,1} e^{(\beta-1)\left\langle \log|f'|\right \rangle} \exp\left\langle\log\left(\alpha-\frac{ff''}{(f')^2}\right)\right\rangle.
\end{equation}
Adopting the notation
\begin{equation}\label{eq:not}
\widetilde p=(\beta-2) p^*,\quad \widetilde \lambda=2-\alpha
\end{equation}
the latter inequality is rewritten as
\begin{equation*}
\phi_{\widetilde p,\widetilde\lambda}[f]^{(\beta-2)\widetilde\lambda}
e^{((\beta-2)\widetilde \lambda+2-2\beta)S[f]}\geqslant |2-\beta|\kappa^{(0)}_{p,1} e^{(\beta-1)\left\langle \log|f'|\right \rangle} \exp\left\langle\log\left(2-\widetilde\lambda-\frac{ff''}{(f')^2}\right)\right\rangle,
\end{equation*}
leading to Eq.~\eqref{eqth:mod_Stam} after dropping the tildes and easy manipulations. The only condition for the previous estimates to hold true is $p^*>0$, which is equivalent to $\sign(\widetilde p)=\sign(\beta-2)$, completing the proof in the case $\beta\neq2$. Letting now $\beta=2,$ we infer from Eq. \eqref{eq:Sp_down2} that
\begin{equation}\label{eq:downdown2_sigma}
\sigma_{p^*}[f^{\downarrow\downarrow}_{\alpha2}]=\left[\int_{\Rset} \adown(s) |\log\adown(s)|^{p^*} ds\right]^{\frac{1}{p^*}}
=\left[\int_{\Rset} f(x) \left|\log\frac{f(x)^\alpha}{|f'(x)|}\right|^{p^*} dx\right]^{\frac{1}{p^*}},
\end{equation}
which, together with Eq.~\eqref{eqprop:downdown}, gives Eq.~\eqref{eqth:mod_Stam2} after introducing again the notation $\lambda=2-\alpha$. The minimizing density must satisfy $\down{\beta}\down{\alpha}[f^{\rm min}]=g_{p,1};$ in the case $\beta\neq 2$, recalling the notation \eqref{eq:not} and noticing that
$$
p=\left(\frac{\widetilde p}{\beta-2}\right)^*=\frac{\widetilde p}{2+\widetilde p-\beta},.
$$
the proof is completed for the inequality~\eqref{eqth:mod_Stam}. When $\beta=2$, the minimizing density simply satisfies $\down{2}\down{2-\lambda}[f_{min}]=g_{p,1}$, finishing the proof for Eq.~\eqref{eqth:mod_Stam2}.

In order to prove the inequality~\eqref{eqth:subF-further} we start from the moment-entropy inequality for $(p,\lambda)$ satisfying the conditions ~\eqref{eq:param_clas} or~\eqref{eq:param_mir}, that is,
$$
\left(\frac{\sigma_{p^*}[f]}{N_{\lambda}[f]}\right)^{\theta_0(p,\lambda)}\geqslant \kappa^{(0)}_{p,\lambda},
$$
where $\theta_0(p,\lambda)$ is defined in Eq. \eqref{ineq:bip_E-M_compact}. Assuming that $\sup\left(\frac{f f''}{(f')^2}\right)<\alpha$, we can apply the latter inequality to $f^{\downarrow\downarrow}_{\alpha2}$, letting again $\lambda=2-\alpha$. Using Eq.~\eqref{eq:downdown2_sigma} and noticing that Eq. \eqref{eq:Ren_down2} and Lemma~\ref{lem:sub-fisher} give
$$
N_\lambda[f^{\downarrow\downarrow}_{\alpha2}]=F_{1-\lambda,0}^{\frac1{1-\lambda}}[\adown]=\varphi_{1-\lambda,1-\lambda,0}^\frac{1}{1-\lambda}[f],
$$
we readily arrive to the inequality~\eqref{eqth:subF-further}, whose minimizing densities must satisfy $(f_{\rm min})^{\downarrow\downarrow}_{\alpha2}=\widetilde g_{p,\lambda}$, or equivalently $f_{\rm min}=\up{\alpha}\up2[\widetilde g_{p,\lambda}]$.
\end{proof}
The case $p=2,\lambda=1$ (which imposes $\beta>2$) in inequality~\eqref{eqth:mod_Stam} looks especially remarkable, since it connects the standard Fisher information with the Shannon entropy, leading to a remarkable extension of the classical Stam inequality. To this end, we require the extra condition
\begin{equation}\label{eq:cond_Stam}
\lambda\omega(\beta,\lambda)<0, \quad {\rm that \ is}, \quad \lambda<1+\frac{\beta}{\beta-2}.
\end{equation}
Introducing also the notation
$$
H_{\lambda,\beta} (f)=|f'|^{\frac{\beta-1}{(\beta-2)\lambda+2-2\beta}}\,\left(2-\lambda-\frac{ff''}{(f')^2}\right)^\frac1{(\beta-2)\lambda+2-2\beta},
$$
we have
\begin{corollary}
Let $\beta>2$ and $\lambda$ be such that the condition \eqref{eq:cond_Stam} is in force. Then one can write
\begin{equation}\label{eq:corStam}
\kappa_{p,\lambda,1}^{(1)}
\leqslant
\phi_{p,\lambda}[f]e^{S[f]}
\leqslant
\mathscr E\,\phi_{p,\lambda}[f]^{\frac{2-2\beta}{(\beta-2)\lambda+2-2\beta}}\, e^{\langle\log\left(H_{\lambda,\beta}(f)\right)\rangle}, \quad \mathscr E= \left(\widetilde{\kappa}_{p,\beta}\right)^{\frac1{(\beta-2)\lambda+2-2\beta}}.
\end{equation}
In the classical case $p=2,\lambda=1$ the latter inequality writes
\begin{equation}
\kappa_{2,1}^{(1)}
\leqslant
\sqrt{F[f]}e^{S[f]}
\leqslant
\mathscr E\,F[f]^{\frac{\beta-1}{\beta}}\, e^{\langle\log\left(H_{1,\beta}(f)\right)\rangle},\quad \beta>2.
\end{equation}
\end{corollary}
\begin{proof}
It follows from Eq. \eqref{eqth:mod_Stam} that
$$
e^{S[f]}\leqslant\mathscr E\,\phi_{p,\lambda}[f]^\frac{-1}{\omega(\beta,\lambda)}\, e^{\langle\log\left(H_{\lambda,\beta}(f)\right)\rangle},
$$
where $\mathscr E$ is defined in Eq. \eqref{eq:corStam}. Consequently,
$$\kappa_{p,\lambda,1}^{(1)}
\leqslant
\phi_{p,\lambda}[f]e^{S[f]}
\leqslant
\mathscr E\,\phi_{p,\lambda}[f]^{1+\frac{-1}{\omega(\beta,\lambda)}}\, e^{\langle\log\left(H_{\lambda,\beta}(f)\right)\rangle},
$$
which is equivalent to Eq. \eqref{eq:corStam}. Note that $\beta>2$ ensures $\omega(\beta,1)<0$, which implies that the condition \eqref{eq:cond_Stam} is fulfilled when $\lambda=1$.
\end{proof}

\section{Conclusions and open problems}

In this paper, we have introduced a new class of informational functionals that we have called upper-moments, motivated by a successive application of the up transformation a finite number of times to a probability density function. In a similar manner, we have also introduced the down-Fisher measure, obtained as the Fisher information of a down transformed density. A number of informational inequalities have been established relating these new functionals among them and with the standard informational functionals. We thus conclude that the up and down transformations allow for extending informational inequalities to a wide amount of density functions and domains of parameters, including some that were not considered in the previously established theory, and endow the class of density functions with a rather rich structure of inequalities. The whole family of minimizing densities encompasses a huge diversity including generalized Beta or stretched Gaussians, among others.

More precisely, these informational inequalities can be classified in two wide families:

	$\bullet$ those in which the involved functionals can be related by a single application of an up or down transformation, such as, the classical entropy-moment and  Stam inequalities or the novel upper-moment--moment and down-Fisher--Fisher inequalities.

	$\bullet$ those in which the involved functionals can be related by a pair of applications of up or down transformations, such as, the classical Cramér-Rao inequality, the upper-moment--entropy inequality or the down-Fisher--entropy inequality.

A relevant difference between both families is that, while the "one-step" family of inequalities is satisfied in both classical and mirrored domains of exponents, the "two-step" family of inequalities only holds true in the classical one.

It is worth mentioning that the up and down transformations reveal that the upper-moments and the moments are related in the same manner as the moments and the Rényi entropy power are, and in turn, also in the same manner as the Rényi entropy power and the generalized Fisher information are. This fact is also true regarding the Fisher and the down-Fisher measures.

We employ the previous structure of inequalities in order to derive sharp upper bounds for classical and generalized products between informational functionals. In particular, we associate densities satisfying certain regularity conditions to the existence of sharp upper bounds for the main products between functionals, such as the moment-entropy, the Stam-like and the Cramér-Rao-like quantities.
We expect that the upper-moments and the down-Fisher measure introduced in this work can find more practical applications in the future.

\section*{Appendix}

This technical section is dedicated to giving more explicit forms of the minimizers $\widetilde g^\bigstar_{p,\alpha,q}$ introduced in Eq. \eqref{def:gbigstar}. All the expressions given below follow easily from the general ones given in \cite[Proposition 4.1]{IP25}, to which we refer the interested reader for a proof. We have a number of different cases depending on the parameters $q$ and $\alpha$, as follows:
\begin{enumerate}[a)]
	\item For $q\neq 0$ and $\alpha\neq 2,$ we have
	\begin{equation*}
\begin{split}
 g^\bigstar_{p,\alpha,q}&=\down{\alpha}[g_{p,1+(2-\alpha)q}](s)\\
 &=\frac{|(2-\alpha) q|^{\frac 1p}}{p ^*\;a_{p,1+(2-\alpha)q}^{\frac{(2-\alpha)q}{p ^*}}}\left[(\alpha-2)s\right]^{\frac{1}{2-\alpha}+q-1}\left|\left[(\alpha-2)s\right]^{q}-a_{p,1+(2-\alpha)q}^{(2-\alpha)q}\right|^{-\frac {1}p}.
 \end{split}
\end{equation*}
and
\begin{equation*}
\begin{split}
\overline g^\bigstar_{p,\alpha,q}&=\down{\alpha}[\overline  g_{p,1+(2-\alpha)q}](s)=\down{\alpha}[g_{(\alpha-2)q,1-p}](s)\\
&=\mathcal K
\left[(\alpha-2)s\right]^{\frac{1-p}{2-\alpha}-1}\left|\left[(\alpha-2)s\right]^{\frac{p}{\alpha-2}}-a_{(\alpha-2)q,1-p}^{-p}\right|^{\frac {1}{(2-\alpha)q}}
\end{split}
\end{equation*}
with
\begin{equation*}
\begin{split}
\mathcal K=\frac{|p|^{\frac 1{(\alpha-2)q}}((\alpha-2)q-1)}{a_{(\alpha-2)q,1-p}^{\frac{(1+(2-\alpha)q)p}{(\alpha-2)q}}(\alpha-2)q}
&=|p|^{\frac 1{(\alpha-2)q}}a_{(\alpha-2)q,1-p}^{p\left(1-\frac1{(\alpha-2)q}\right)}\left(1-\frac1{(\alpha-2)q}\right)\\
&=\mathcal C|p|^{1-\mathcal C}a_{(\alpha-2)q,1-p}^{p\mathcal C}
\end{split}
\end{equation*}
where $ \mathcal C=1-\frac1{(\alpha-2)q}.$

\item For $q=0$ and $\alpha\neq 2,$
	\begin{equation*}
g^\bigstar_{p,\alpha,0}=\down{\alpha}[g_{p,1}](s)=\frac{1}{p^*}[(\alpha-2) s]^{\frac{\alpha-1}{2-\alpha}}\left(\frac{\ln[a_{p,1}^{\alpha-2}(\alpha-2) s]}{\alpha-2}\right)^{-\frac{1}{p}}.
	\end{equation*}
	\item For $\lambda\neq 1$ and $\alpha=2$, (that is, $q=\infty$),
	\begin{equation*}
g^\bigstar_{p,2,q}=\down{2}[g_{p,\lambda}]=\frac{|1-\lambda|^\frac 1p a_{p,\lambda}^{\frac 1{1-\lambda}}}{p^*} e^{-\lambda s}\left|a_{p,\lambda}^{-1}e^{(1-\lambda) s}-1\right|^{-\frac 1p}.
	\end{equation*}
and
	\begin{equation*}
	\overline g^\bigstar_{p,2,\lambda}=\down{2}[\overline  g_{p,\lambda}](s)=\down{2}[g_{1-\lambda,1-p}](s)
	=\frac{|p|^\frac 1{1-\lambda} a_{1-\lambda,1-p}^{\frac 1{p}}}{(1-\lambda)^*} e^{(p-1)s}\left|a_{1-\lambda,1-p}^{-1}e^{p s}-1\right|^{\frac 1{\lambda-1}}.
\end{equation*}	

	\item For $\lambda=1$ and $\alpha=2$ (that is, $q=\infty$),
	\begin{equation*}
g^\bigstar_{p,2,1}=\down{2}[g_{p,1}]=\frac{e^{-s}}{p ^* (s+\ln a_{p,1})^\frac 1p}.
	\end{equation*}
\end{enumerate}
In all the previous expressions, the constants $a_{p,\lambda}$ are the normalization constants of the stretched deformed Gaussian densities $g_{p,\lambda}$, see \eqref{def:g_plambda}.

\subsection*{Acknowledgements}
	

R. G. I. and D. P.-C. are very grateful to the Ministry of Science and Innovation of Spain for its partial financial support under under Grants PID2020-115273GB-100 (AEI/FEDER, EU) and RED2022-134301-T funded by \text{MCIN/AEI/10.13039/501100011033}. D. P.-C. is also partially supported by the Grant PID2023-153035NB-100 funded by MICIU/AEI/10.13039/501100011033 and “ERDF/EU A way of making Europe”.

\bigskip

\noindent \textbf{Data availability} Our manuscript has no associated data.

\bigskip

\noindent \textbf{Competing interest} The authors declare that there is no competing interest.

\bibliographystyle{unsrt}
\bibliography{refs}
\end{document}